\documentclass[twocolumn,twoside]{IEEEtran} 

\makeatletter
\let\NAT@parse\undefined
\makeatother

\usepackage{amsmath,amssymb,amsfonts}
\usepackage[final]{graphicx}
\usepackage{flushend}
\usepackage[numbers,sort&compress]{natbib}

\usepackage{ucs} 
\usepackage[utf8x]{inputenc}
\usepackage[usenames,dvipsnames]{xcolor}
\usepackage{tikz}
\usepackage{tkz-tab}
\usepackage{multirow}
\usepackage{latexsym}
\usepackage{amssymb}
\usepackage{amsmath}

\newtheorem{corollary}{Corollary}

\newtheorem{lemma}{Lemma}

\def \L {{\mathcal L }}

\makeatletter
\def\blfootnote{\xdef\@thefnmark{}\@footnotetext}
\makeatother

\begin{document}
\title{\Huge The Fluctuating Two-Ray Fading Model: Statistical Characterization and Performance Analysis}

\author{
\vspace{3mm}
\authorblockN{ Juan M. Romero-Jerez, F. Javier Lopez-Martinez, Jos\'e F. Paris and Andrea J. Goldsmith}}


\maketitle
\begin{abstract}
\blfootnote{\noindent This work will be presented in part at IEEE Globecom 2016. J. M. Romero-Jerez is  with Departmento de Tecnolog\'ia Eectr\'onica, Universidad de Malaga - Campus de Excelencia Internacional Andaluc\'ia Tech., Malaga 29071, Spain. Contact e-mail: romero@dte.uma.es. F. J. Lopez-Martinez and J. F. Paris are with Departmento de Ingenier\'ia de Comunicaciones, Universidad de Malaga - Campus de Excelencia Internacional Andaluc\'ia Tech., Malaga 29071, Spain. A. Goldsmith is with the Wireless Systems Lab, Department of Electrical Engineering, Stanford University, CA, USA.
\\
\indent This work has been funded by the Consejer\'ia de Econom\'ia, Innovaci\'on, Ciencia y Empleo of the Junta de Andaluc\'ia, the Spanish Government
and the European Fund for Regional Development FEDER (projects P2011-TIC-7109, P2011-TIC-8238, TEC2014-57901-R,  TEC2013-42711-R and TEC2013-44442-P
\\
\indent This work has been submitted to the IEEE for publication. Copyright may be transferred without notice, after which this version may no longer be accessible.
}
We introduce the Fluctuating Two-Ray (FTR) fading model, a new statistical channel model that consists of two fluctuating specular components with random phases plus a diffuse component. The FTR model arises as the natural generalization of the two-wave with diffuse power (TWDP) fading model; this generalization allows its two specular components to exhibit a random amplitude fluctuation. Unlike the TWDP model, all the chief probability functions of the FTR fading model (PDF, CDF and MGF) are expressed in closed-form, having a functional form similar to other state-of-the-art fading models. We also provide approximate closed-form expressions for the PDF and CDF in terms of a finite number of elementary functions, which allow for a simple evaluation of these statistics to an arbitrary level of precision. We show that the FTR fading model provides a much better fit than Rician fading for recent small-scale fading measurements in 28 GHz outdoor millimeter-wave channels. Finally, the performance of wireless communication systems over FTR fading is evaluated in terms of the bit error rate and the outage capacity, and the interplay between the FTR fading model parameters and the system performance is discussed. Monte Carlo simulations have been carried out in order to validate the obtained theoretical expressions.
\end{abstract}

\vspace{0mm}
\begin{keywords}
Wireless channel modeling, envelope statistics, moment generating function, multipath propagation, Rician fading, small-scale fading, two-ray.
\end{keywords}

\IEEEpeerreviewmaketitle

\section{Introduction}
\label{intro}
The use of millimeter-wave (mmWave) bands to overcome the wireless spectrum shortage caused by the exponential increase in aggregate traffic is being embraced by emerging wireless standards such as 5G \cite{Boccardi2014}. This has led to significant research in mmWave radio communications in urban outdoor environments \cite{Rappaport2013}. 
Much of this research has focused on mmWave channel modeling \cite{Akdeniz2014,Metis2014,Rappaport2015,Hur2016}. 

Most of the stochastic channel models for mmWave communications assume Rayleigh or Rician distributions for the small-scale fading path amplitudes in NLOS and LOS scenarios, respectively. Very recently \cite{Samimi2016}, the small-scale fading statistics obtained from a 28 GHz outdoor measurement campaign showed that Rician fading was more suited than Rayleigh even in NLOS environments. However, a deeper look into the results of \cite{Samimi2016} indicates that conventional fading models in the literature do not accurately model the random fluctuations suffered by the received signal. In particular, the empirical CDFs and PDFs reported in \cite{Samimi2016} and \cite{Mavridis2015}, respectively, for different mmWave scenarios exhibit a bimodality that cannot be captured even by generalized fading models \cite{Yacoub07,Paris2014,Cotton2015}. 

We here propose a new amplitude fading model, the Fluctuating Two-Ray (FTR) fading model, whose statistical distribution captures the wide heterogeneity of random fluctuations a signal experiences in propagation environments with multiple scatterers. The FTR model is a natural generalization\footnote{We also note that the FTR model here proposed differs from the Generalized Two-Ray (GTR) model proposed in \cite{Rao2015}. Unlike the TWDP model, the GTR model allows the phase distributions of the specular waves to be other than uniform, but the amplitudes of the specular components are still kept constant.} of the two-wave with diffuse power (TWDP) fading model proposed by Durgin, Rappaport and de Wolf \cite{Durgin2002}. In this generalization, the constant-amplitude specular waves randomly fluctuate. The inclusion of an additional source of randomness allows for a better characterization of the amplitude fluctuations experienced by the radio signal, compared to the TWDP model (which is indeed included as a particular case of the FTR model). Remarkably, this larger flexibility does not come at the price of an increased mathematical complexity, but instead facilitates the analytical characterization of this new fading model.

The benefits of using the FTR fading model, which will be derived below, can be summarized as follows: (1) Despite being more general than the original TWDP model, the primary probability functions (CDF, PDF and MGF) of the FTR model are given in closed-form, (2) The FTR fading distribution is inherently bimodal, but also includes classical unimodal fading models like Rician, Nakagami-$m$, Hoyt and Rayleigh as particular cases; thus, it can be matched to a wider variety of propagation conditions than conventional fading models, and (3) The FTR fading distribution provides a much better fit than existing fading models to field measurements, for example the 28 GHz field measurements recently reported in \cite{Samimi2016}.

The remainder of this paper is structured as follows: the physical justification of the FTR fading model is introduced in Section \ref{systemmodel}. Then, in Section \ref{analysis}, the FTR fading model is statistically characterized in terms of its MGF, CDF and PDF. The empirical validation of our model is presented in Section \ref{fit} by fitting the FTR fading model to small-scale fading field measurements in the mmWave band. The performance of wireless communication systems operating under FTR fading is analyzed in Section \ref{performance}, with associated numerical results given in Section \ref{numerical}. Our main conclusions are summarized in Section \ref{conc}.

\section{Preliminaries and channel model}
\label{systemmodel}

The small-scale fluctuations in the amplitude of a signal transmitted over a wireless channel can be modeled by the superposition of a set of $N$ dominant waves, referred to as specular components, to which other diffusely propagating waves are added \cite{Durgin2002}. Under this model, the complex baseband voltage of a wireless channel experiencing multipath fading can be expressed as 
\begin{equation}
\label{eq:01}
V_r=\sum\limits_{n = 1}^N  V_n \exp \left( {j\phi _n } \right) + X + jY,
\end{equation}
where $ V_n \exp \left( {j\phi _n } \right)$ represents the \emph{n-th} specular component, which is assumed to have a constant amplitude $V_n$ and a uniformly distributed random phase $\phi _n $, such that $\phi _n \sim \mathcal{U}[0,2\pi)$. Since the distances traversed by the propagating waves are typically orders of magnitude greater than their wavelengths, the random phase variables of each specular component are assumed to be statistically independent. On the other hand, $X + jY$ is a complex Gaussian random variable, such that $X,Y \sim \mathcal{N}(0,\sigma^2)$, representing the diffuse received signal component due to the combined reception of numerous weak, independently-phased scattered waves. This Gaussian model is based on application of the central limit theorem to the sum of these numerous waves.

The general model presented in (\ref{eq:01}) includes very important statistical wireless channel models as particular cases. Thus, when $N=0$, i.e., no specular component is present, the Rayleigh fading model is obtained, while for  $N=1$, a single dominant specular component, we have the Rician fading model. The case in which there are two dominant specular components ($N=2$) is usually referred to as the Two Wave with Diffuse Power (TWDP) fading model or, alternatively, the Generalized Two-Ray fading model with Uniformly distributed phases (GTR-U) \cite{Rao2015}. This recently-developed model contains the aforementioned classical fading models as particular cases and accurately fits field measurements in a variety of propagation scenarios \cite{Durgin2002}. Unfortunately, the statistical characterization of the TWDP fading model is much more complicated than that of classical fading models, as there are not known closed-form expressions for the PDF and the CDF of the received signal envelope. Notably, the MGF of the power envelope in TWDP fading was recently derived in \cite{Rao2015}.

The specular components in the general model in (\ref{eq:01}) have constant amplitudes. We must here note that variations in the amplitude of the dominant specular components, often associated with the LOS propagation, have been considered in some specific scenarios and validated with field measurements: these are the cases of the Rician shadowed fading model \cite{Abdi2003} which generalizes the Rician fading model, or the $\kappa$-$\mu$ shadowed fading model introduced in \cite{Paris2014} as a generalization of Yacoub's $\kappa$-$\mu$ fading model. However, while the word ``shadowing'' was used when the models \cite{Abdi2003,Paris2014} were introduced, these models should not necessarily be linked to the large-scale fading phenomena also called shadowing, due to a complete or partial blockage by obstacles many times larger than the signal wavelength. Instead, these models reflect any amplitude fluctuation in the specular waves (e.g. say variations in the propagation condition or fast moving scatterers) that takes place over the time period of interest.
Therefore, considering the amplitudes of the specular components to be modulated by a Nakagami-\emph{m} random variable with squared unit mean as in \cite{Abdi2003,Paris2014}, we can write:
\begin{equation}
\label{eq:02}
V_r=\sum\limits_{n = 1}^N  \sqrt{\zeta} V_n \exp \left( {j\phi _n } \right) + X + jY,
\end{equation}
where $\zeta$ is a unit-mean Gamma distributed random variable with PDF
\begin{equation}
\label{eq:03}
f_\zeta  \left( u \right) = \frac{{m^m u^{m - 1} }}
{{\Gamma \left( m \right)}}e^{ - mu} .
\end{equation}
Note that we are considering the same fluctuation for the specular components, which is actually a natural situation in different wireless scenarios. When the scatterers are in the vicinity of the transmitter and/or the receiver,  the specular components will travel alongside most of the way, and the eventual channel fluctuations would affect  them simultaneously. This is, for example, the case of the human body shadowing as the user moves. Also, there is a number of causes of electromagnetic disturbances that will typically affect the speculat components simultaneously, including ionospheric scintillation (for satellite communications), sudden changes of the channel electromagnetic field  due to natural (e.g. solar activity) or artificial (e.g. motors ignition, electric power generators) sources, etc.

The wireless channel model given in (\ref{eq:02})-(\ref{eq:03}) for the particular case when ${N=1}$ corresponds to the Rician shadowed fading model \cite{Abdi2003}. In the rest of this paper, we will consider the case when ${N=2}$ and will derive a statistical description of the resulting channel model. This model will be subsequently denoted as the  Fluctuating Two-Ray (FTR) model, in order to indicate the presence of two specular components with random phase for which their amplitude exhibits a random fluctuation.

\section{Statistical characterization of the FTR fading model}
\label{analysis}

Let us consider the complex baseband received signal, which can be written as
\begin{equation}
\label{eq:04}
V_r  = \sqrt \zeta V_1  \exp \left( {j\phi _1 } \right) + \sqrt \zeta V_2  \exp \left( {j\phi _2 } \right) + X + jY.
\end{equation}
This model is conveniently expressed in terms of the parameters $K$ and $\Delta$, defined as
\begin{equation}
\label{eq:05}
K = \frac{{V_1^2  + V_2^2 }}
{{2\sigma ^2 }},
\end{equation}
\begin{equation}
\label{eq:06}
\Delta  = \frac{{2V_1 V_2 }}
{{V_1^2  + V_2^2 }}.
\end{equation}
The $K$ parameter represents the ratio of the average power of the dominant components to the power
of the remaining diffuse multipath, just like the Rician $K$ parameter. On the other hand, $\Delta$ is a parameter ranging from 0 to 1 expressing how similar to each other are the average received powers of the specular components: when the magnitudes of the two specular components are equal, $\Delta=1$ , while in the absence of a second component
($V_1=0$ or $V_2=0$), $\Delta=0$. Note that $\Delta=0$ yields the Rician shadowed fading model.

We will first characterize the distribution of the received power envelope associated with the FTR fading model, or equivalently, the distribution of the received signal-to-noise ratio (SNR). After passing through the multipath fading channel, the signal will be affected by additive white Gaussian noise (AWGN) with one-sided power spectral density $N_0$. The statistical characterization of the instantaneous SNR, here denoted as $\gamma$, is crucial for the analysis and design of wireless communications systems, as many performance metrics in wireless communications are a function of the SNR.

The received average SNR $\bar\gamma$ after transmitting a symbol with energy density $E_s$ undergoing a multipath fading channel as described in (\ref{eq:04}) will be
\begin{equation}
\begin{split}
\label{eq:07}
  & \bar \gamma  = \left( {E_b /N_0 } \right){\mathbb E}\left\{ {\left| {V_r } \right|^2 } \right\} = \left( {E_b /N_0 } \right)\left( {V_1^2  + V_2^2  + 2\sigma ^2 } \right)  \cr 
  & \quad \quad  = \left( {E_b /N_0 } \right)2\sigma ^2 \left( {1 + K} \right), \cr
\end{split}
\end{equation}
where $\mathbb E\{\cdot\}$ denotes the expectation operator.

With all the above definitions, the chief probability functions related to the FTR fading model can now be computed.

\subsection{MGF}
In the following lemma we show that, for the FTR fading model, it is possible to obtain the MGF of $\gamma$ in closed-form.

\begin{lemma}  
Let us consider the FTR fading model as described in (\ref{eq:04})-(\ref{eq:07}). Then, the MGF of the received SNR $\gamma$ will be given by
\begin{equation}
\label{eq:08}
\begin{split}
  & M_\gamma  \left( s \right) = \frac{{m^m \left( {1 + K} \right)\left( {1 + K - \bar \gamma s} \right)^{m - 1} }}
{{\left( {\sqrt {\mathcal{R}\left( {m,k,\Delta ;s} \right)} } \right)^m }}  \cr 
  & \quad \quad \quad \quad \times P_{m - 1} \left( {\frac{{m\left( {1 + K} \right) - \left( {m + K} \right)\bar \gamma s}}
{{\sqrt {\mathcal{R}\left( {m,k,\Delta ;s} \right)} }}} \right), \cr
\end{split}
\end{equation}
where $\mathcal{R}\left( {m,k,\Delta ;s} \right)$ is a polynomial in $s$ defined as

\begin{equation}
\label{eq:09}
\begin{split}
  & R\left( {m,k,\Delta ;s} \right) = \left[ {\left( {m + K} \right)^2  - \Delta ^2 K^2 } \right]\bar \gamma ^2 s^2   \cr 
  & \quad \quad \quad  \quad  - 2m\left( {1 + K} \right)\left( {m + K} \right)\bar \gamma s + m^2 \left( {1 + K} \right)^2,  \cr 
\end{split}
\end{equation}
and  $P_\mu (\cdot$) is the Legendre function of the first kind of degree $\mu$, which can be calculated as
\begin{equation}
\label{eq:10}
P_\mu  \left( z \right) = _2 F_1 \left( { - \mu ,\mu  + 1;1;\frac{{1 - z}}
{2}} \right),
\end{equation}
given that
\begin{equation} \label{eq:11}
\begin{split}
\left| {1 - z} \right| < 2,
\end{split}\end{equation}
where $_2F_1 (\cdot)$ is the Gauss hypergeometric function \cite[p. 556 (15.1.1)]{Abramowitz72}.

\end{lemma}

\begin{proof} See Appendix \ref{App1}.
\end{proof}

The FTR fading model introduced here is well-suited to recreate the propagation conditions in a wide variety of wireless scenarios, ranging from very favorable ones to worse-than Rayleigh fading. It also includes many important well-known statistical fading models as particular cases, i.e., TWDP, Rician shadowed, Rician, Rayleigh, one-sided Gaussian, Nakagami-\emph{m} and Nakagami-\emph{q} (Hoyt). The connection between the FTR fading model and the special cases included therein can easily be validated using the previous definitions for $K$, $\Delta$ and $m$, and is formally stated in Table \ref{table_1}.

\begin{table}[!t]
\renewcommand{\arraystretch}{1.7}
\centering
\caption{\textsc{Connections between the FTR fading and other fading models in the literature. The FTR fading parameters are underlined to avoid confusion with the special cases}.}
\label{table_1}
\centering
\begin{tabular}
{c|c}
\hline
\hline
Channels  & FTR Fading Parameters\\
\hline
\hline
One-sided Gaussian &  a) \b{$\Delta$}~$=0$, \b{$K$}~$\rightarrow\infty$, \b{$m$}~$= 0.5$ \\ & b)  \b{$\Delta$}~$=1$, \b{$K$}~$\rightarrow\infty$, \b{$m$}~$= 1$ \\
\hline
\multirow{2}{*}{Rayleigh} &  a) \b{$\Delta$}~$=0$, \b{$K$}~$\rightarrow\infty$, \b{$m$}~$= 1$ \\ & b)  \b{$\Delta$}~$=0$, \b{$K$}~$=0$, $\forall$\b{$m$}~$$ \\
\hline
\multirow{2}{*}{Nakagami-$q$ (Hoyt)} &  a) \b{$\Delta$}~$=0$, \b{$K$}~$=\frac{1-q^2}{2q^2}$, \b{$m$}~$= 0.5$ \\ & b)  $\forall$\{{$\Delta$},{$K$}\}, \text{with} $q=\sqrt{\frac{1+K(1-\Delta)}{1+K(1+\Delta)}}$, \b{$m$}~$=1$ \\
\hline
Nakagami-$m$ &  \b{$\Delta$}~$=0$, \b{$K$}~$\rightarrow\infty$, \b{$m$}~$= m$\\  
\hline
Rician &  \b{$\Delta$}~$=0$, \b{$K$}~$=K$, \b{$m$}~$\rightarrow\infty$\\  
\hline
Rician shadowed &  \b{$\Delta$}~$=0$, \b{$K$}~$=K$, \b{$m$}~$= m$\\  
\hline
TWDP &  \b{$\Delta$}~$=\Delta$, \b{$K$}~$=K$, \b{$m$}~$\rightarrow\infty$\\  
\hline
Two-Wave &  \b{$\Delta$}~$=\Delta$, \b{$K$}~$\rightarrow\infty$, \b{$m$}~$\rightarrow\infty$\\  
\hline
Fluctuating Two-Wave &  \b{$\Delta$}~$=\Delta$, \b{$K$}~$\rightarrow\infty$, \b{$m$}~$=m$\\  
\hline
\hline
\end{tabular}
\end{table}

Special attention is merited for the case of Nakagami-$q$ (Hoyt) fading, which can be seen as a special case of the FTR fading model in two different ways. The first one arises after specializing the Rician shadowed model for $m=1/2$ as indicated in \cite{Laureano2015}; however, as we will later see, choosing the parameter $m$ to be a positive integer has additional benefits in terms of mathematical tractability. Thus, in the following corollary we show how the Nakagami-$q$ (Hoyt) fading model can be obtained from the FTR fading model with $m=1$.

\begin{corollary} \label{C1}
For $m=1$, the FTR fading model becomes the Nakagami-$q$ (Hoyt) model with
\begin{equation}
\label{eq:28}
q = \sqrt {\frac{{1 + K\left( {1 - \Delta } \right)}}
{{1 + K\left( {1 + \Delta } \right)}}} .
\end{equation}
\end{corollary}

\begin{proof}
See Appendix \ref{App2}.
\end{proof}

\begin{figure}[t]
    \includegraphics[width=1\columnwidth]{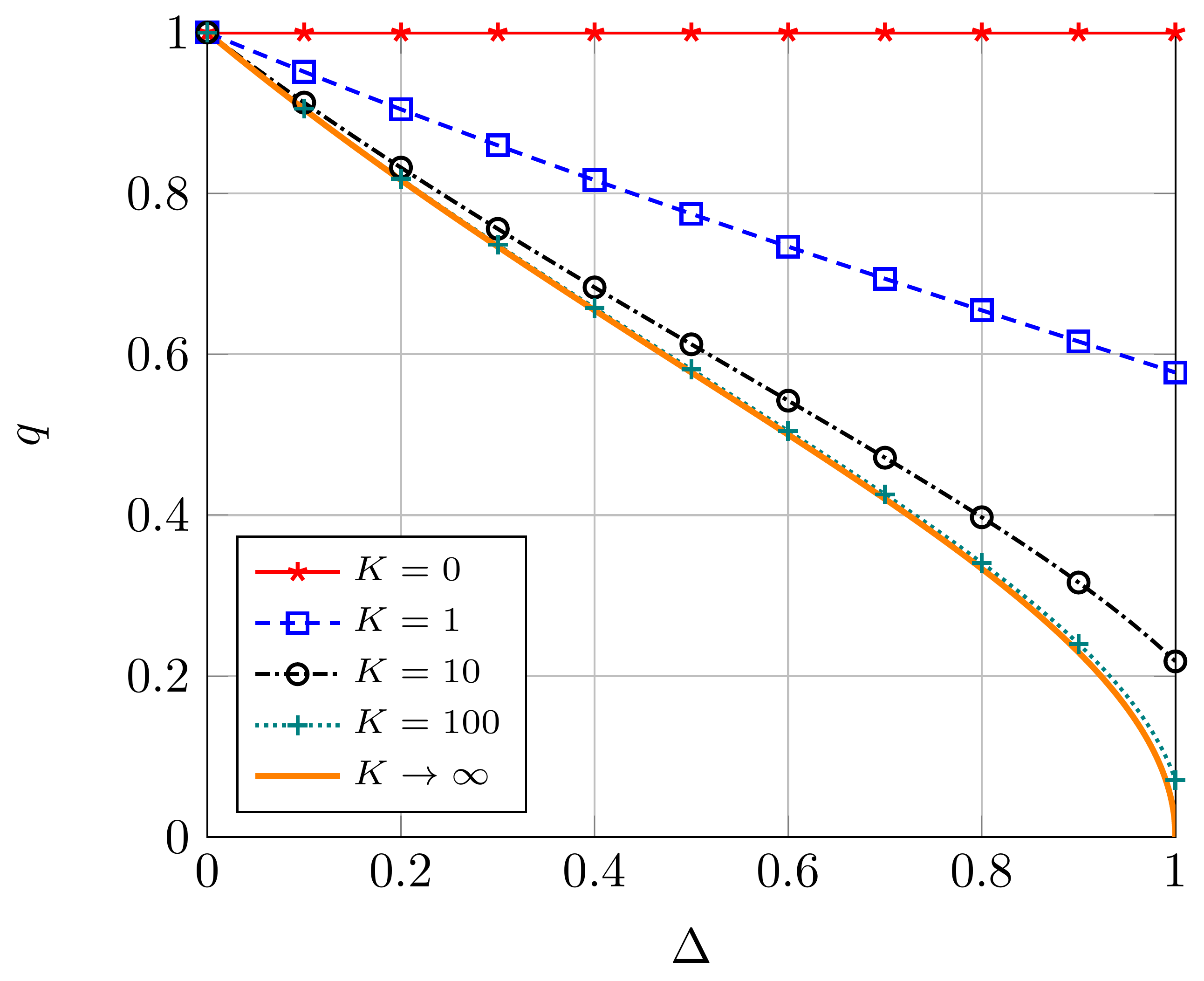}
          \caption{Connection between the FTR and the Nakagami-$q$ fading model parameters, with $m=1$.}
    \label{f0}
  \end{figure}
  
Strikingly, the inherently non-circularly symmetric Hoyt distribution is obtained by adding two specular components with uniformly distributed phases and Rayleigh-distributed random amplitude ($m=1$) to a circularly symmetric diffuse component, for all $q$ satisfying (\ref{eq:28}). Setting $q=0$ or $q=1$ reduces to the one-sided Gaussian and Rayleigh distributions, respectively. Note that the Nakagami-$q$ fading distribution models scenarios worse than Rayleigh (deeper fades). The relationship between $\Delta$, $K$ and $q$ is represented in Fig. \ref{f0}. We see that for low values of $K$, only those values of $q$ closer to 1 are possible for any $\Delta$. As $K$ grows, we observe that the whole range of $q\in[0,1]$ can be attained with $q\rightarrow\sqrt{\frac{1-\Delta}{1+\Delta}}$.

With the MGF in closed-form, we now show that the PDF and CDF of the FTR fading distribution can also be obtained in closed-form, provided that the parameter $m$ is restricted to take positive integer values (i.e., $m \in \mathbb{Z}^{+}$).

\subsection{PDF and CDF}

When the parameter $m$ takes integer values, the MGF of the SNR in the FTR fading model can be calculated as a finite sum of elementary terms. This is based on the fact that, for $m$ an integer, the Legendre function in the MGF given in (\ref{eq:08}) has an integer degree, thus becoming a Legendre polynomial. A Legendre polynomial of degree $n$ can be written as \cite[p. 775 (22.3.8)]{Abramowitz72}
\begin{equation}
\label{eq:34}
P_n \left( z \right) = \frac{1}
{{2^n }}\sum\limits_{q = 0}^{\left\lfloor {n/2} \right\rfloor } {\left( { - 1} \right)^q } C_q^n z^{n - 2q},
\end{equation}
where $\left\lfloor  \cdot  \right\rfloor $ is the floor function and $C_q^n$ is a coefficient given by
\begin{equation}
\label{eq:35}
C_q^n  = \left( {\begin{array}{c}
   n  \\ 
   q  \\ 
 \end{array} } \right)\left( {\begin{array}{c}
   {2n - 2q}  \\ 
   n  \\ 
 \end{array} } \right) = \frac{{\left( {2n - 2q} \right)!}}
{{q!\left( {n - q} \right)!\left( {n - 2q} \right)!}}.
\end{equation}

\begin{figure*}[!t]
\normalsize
\begin{equation}
\label{eq:36}
\begin{split}
  & f_\gamma  \left( x \right) = \frac{{ 1}}
{{2^{m - 1} }}\frac{{1 + K}}
{{\bar \gamma }}\left( {\frac{m}
{{\sqrt {\left( {m + K} \right)^2  - K^2 \Delta ^2 } }}} \right)^m \sum\limits_{q = 0}^{\left\lfloor {(m - 1)/2} \right\rfloor } {\left( { - 1} \right)^q } C_q^{m - 1} \left( {\frac{{m + K}}
{{\sqrt {\left( {m + K} \right)^2  - K^2 \Delta ^2 } }}} \right)^{m - 1 - 2q}   \cr 
  & \quad \quad \quad  \times \;\Phi _2^{(4)} \left( {1 + 2q - m,m - q - \frac{1}
{2},m - q - \frac{1}
{2},1 - m;1;} \right.  \cr 
  & \quad \quad \quad \quad \quad \quad \quad \quad \quad \left. { - \frac{{m\left( {1 + K} \right)}}
{{\left( {m + K} \right)\bar \gamma }}x, - \frac{{m\left( {1 + K} \right)}}
{{\left( {m + K\left( {1 + \Delta } \right)} \right)\bar \gamma }}x, - \frac{{m\left( {1 + K} \right)}}
{{\left( {m + K\left( {1 - \Delta } \right)} \right)\bar \gamma }}x, - \frac{{1 + K}}
{{\bar \gamma }}x} \right). \cr
\end{split}
\end{equation}
\vspace*{4pt}
\end{figure*}

\begin{figure*}[!t]
\normalsize
\begin{equation}
\label{eq:37}
\begin{split}
  & F_\gamma  \left( x \right) = \frac{{ 1}}
{{2^{m - 1} }}\frac{{1 + K}}
{{\bar \gamma }}\left( {\frac{m}
{{\sqrt {\left( {m + K} \right)^2  - K^2 \Delta ^2 } }}} \right)^m \sum\limits_{q = 0}^{\left\lfloor {(m - 1)/2} \right\rfloor } {\left( { - 1} \right)^q } C_q^{m - 1} \left( {\frac{{m + K}}
{{\sqrt {\left( {m + K} \right)^2  - K^2 \Delta ^2 } }}} \right)^{m - 1 - 2q}   \cr 
  & \quad \quad \quad  \times x \;\Phi _2^{(4)} \left( {1 + 2q - m,m - q - \frac{1}
{2},m - q - \frac{1}
{2},1 - m;2;} \right.  \cr 
  & \quad \quad \quad \quad \quad \quad \quad \quad \quad \left. { - \frac{{m\left( {1 + K} \right)}}
{{\left( {m + K} \right)\bar \gamma }}x, - \frac{{m\left( {1 + K} \right)}}
{{\left( {m + K\left( {1 + \Delta } \right)} \right)\bar \gamma }}x, - \frac{{m\left( {1 + K} \right)}}
{{\left( {m + K\left( {1 - \Delta } \right)} \right)\bar \gamma }}x, - \frac{{1 + K}}
{{\bar \gamma }}x} \right). \cr
\end{split}
\end{equation}
\hrulefill
\vspace*{4pt}
\end{figure*}

\begin{figure*}[!t]
\normalsize
\begin{align}
\label{eq:36b}
\widehat f_{\gamma}(x)&\approx \sum_{i=1}^{M}\frac{\alpha_i}{2}\left\{\mathcal{G}_m\left(x;\beta,K(1-\delta_i)\right)+\mathcal{G}_m\left(x;\beta,K(1+\delta_i)\right)\right\},\\
\label{eq:36c}
\mathcal{G}_m\left(x;\beta,K\right)&=\left(\frac{m}{K+m}\right)^m \beta e^{-\beta x\frac{m}{K+m}} \sum_{n=0}^{m-1}\binom{m-1}{n}\left(\frac{K \beta x}{K+m}\right)^n\frac{1}{n!};\hspace{5mm}
%
%
\end{align}
\begin{align}
\label{eq:37b}
\widehat F_{\gamma}(x)&\approx 1- \sum_{i=1}^{M}\frac{\alpha_i}{2} \left\{\mathcal{H}_m\left(x;\beta,K(1-\delta_i)\right)+\mathcal{H}_m\left(x;\beta,K(1+\delta_i)\right)\right\}\\
\label{eq:37c}
\mathcal{H}_m\left(x;\beta,K\right)&=\sum_{n=0}^{m-1}\sum_{j=0}^{m-n-1}\frac{m^{m-n-1}K^{n+j}}{(m+K)^{m-1+j}}\beta^j\frac{(m-n-j)_{n+j}}{(n+j)! j!}e^{-\beta\frac{m}{m+K}x} x^j
\end{align}
\hrulefill
\hrulefill
\vspace*{4pt}
\end{figure*}

We will make use of (\ref{eq:34}) to compute closed-form expressions for the PDF and CDF of the power envelope for the FTR fading model in (\ref{eq:36}) and (\ref{eq:37}), respectively, which will be demonstrated in the next lemma. Note that the PDF and CDF of the received signal envelope can be easily derived from (\ref{eq:36}) and (\ref{eq:37}) by a simple change of variables. Specifically, through a change of variables we get $f_r(r)=2r f_{\gamma}(r^2)$ and $F_r(r)=F_{\gamma}(r^2)$, with $\bar\gamma$ in (\ref{eq:36}) and (\ref{eq:37}) replaced by $\Omega=E\{r^2\}$. 

\begin{lemma}
When $m\in\mathbb{Z}^+$, the PDF and CDF of the SNR $\gamma$ in a FTR fading channel can be expressed in terms of the confluent hypergeometric function $\Phi_2(\cdot)$
defined in \cite[p. 34, (8)]{Srivastava1985}, as given, respectively, in (\ref{eq:36}) and (\ref{eq:37}).
\end{lemma}

\begin{proof}
See Appendix \ref{App3}.
\end{proof}

Note that despite requiring the evaluation of a confluent hypergeometric function, the PDF and CDF of the FTR fading model can be expressed in terms of a well-known function in communication theory. In fact, the $\Phi_2$ function also makes an appearance in the CDF of common fading models such as Rician shadowed or $\kappa$-$\mu$ shadowed \cite{Paris2010,Paris2014}. Moreover, this function can be efficiently evaluated using an inverse Laplace transform as described in \cite[Appendix 9B]{AlouiniBook}. Thus, the evaluation of the FTR distribution functions does not pose any additional challenge compared to other state-of-the-art fading models.

In the following lemma, we also present a family of approximate PDFs and CDFs for the FTR fading model, which are given in terms of a finite sum of exponential functions and powers. Thus, its evaluation becomes as simple as evaluating the well-known Gamma distribution associated with the squared envelope in the Nakagami-$m$ fading model.

\begin{lemma}
When $m\in\mathbb{Z}^+$, the PDF and CDF of the SNR $\gamma$ in a FTR fading channel can be approximated by a finite sum of elementary functions, as given, in (\ref{eq:36b}) and (\ref{eq:37b}) respectively, where $M>\lceil K \Delta \rceil$, $\beta=\frac{K+1}{\bar\gamma}$ and the coefficients $\alpha_i$ and $\delta_i$ are defined in (\ref{coef1}) and (\ref{coef2}) in the Appendix \ref{App4}.
\end{lemma}

\begin{proof}
See Appendix \ref{App4}.
\end{proof}

In the next set of figures (Figs. 2 to 7), we study the effect of the FTR fading model parameters $K$, $\Delta$ and $m$ on the shape of the PDF. Specifically, the received signal envelope PDF $f_r(r)$ and the power envelope $f_{\gamma}(\gamma)$ are represented in order to better illustrate the versatility of the FTR fading model. 
For the approximated results, $M=\lceil K \Delta \rceil+1$ has been considered in every case.
Monte Carlo simulations have been carried out in order to validate the depicted functions, but they are not represented in these figures as the simulated values are indistinguishable from the exact results.
Similarly to the TWDP fading model, the FTR fading model is inherently bimodal; such bimodality is dominated by the parameters $K$ and $\Delta$. Specifically, $\Delta\rightarrow 1$ and large values of $K$ yield a more pronounced bimodality; this corresponds to the worse-than-Rayleigh fading case. The additional parameter $m$ smoothens such bimodality\footnote{The bimodality of the distribution is clearly identified by the appearance of two maxima in its PDF; this would be translated into several transitions from concavity to convexity (i.e., inflection points) in the CDF.
} as $m$ decreases; conversely, as $m\to\infty$, the FTR fading model reduces to the TWDP fading model. 

{
\begin{figure}[ht]
    \includegraphics[width=0.99\columnwidth]{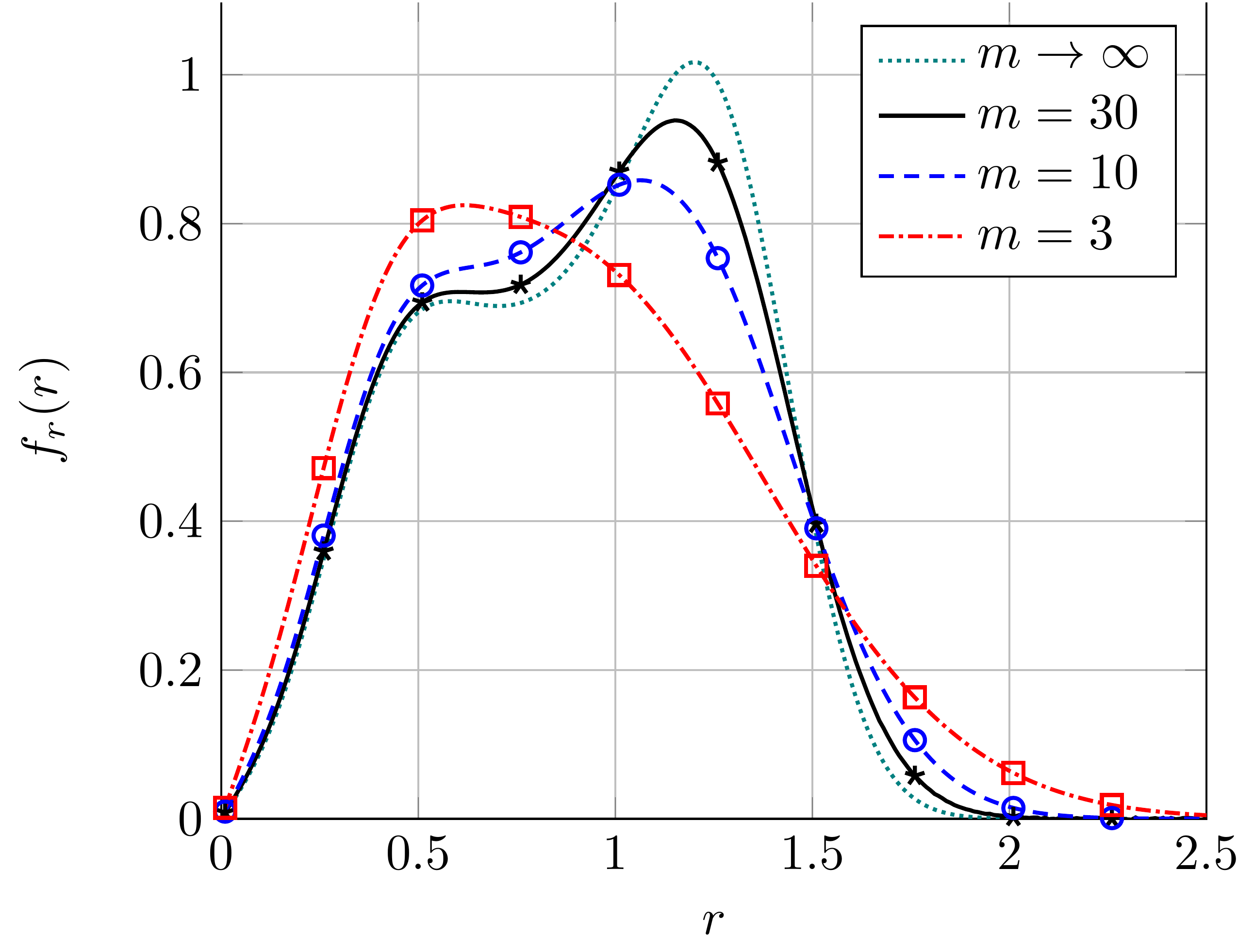}
          \caption{FTR signal envelope distribution for different values of $m$, with $K=15$, $\Delta=0.9$ and $\Omega=1$. Solid lines correspond to the exact PDF derived from (\ref{eq:36}), markers correspond to the approximate PDF derived from (\ref{eq:36b}). The case $m\rightarrow\infty$ reduces to the TWDP fading distribution \cite{Durgin2002}.}
    \label{fig:subfig1}
  \end{figure}

\begin{figure}[ht]
    \includegraphics[width=0.99\columnwidth]{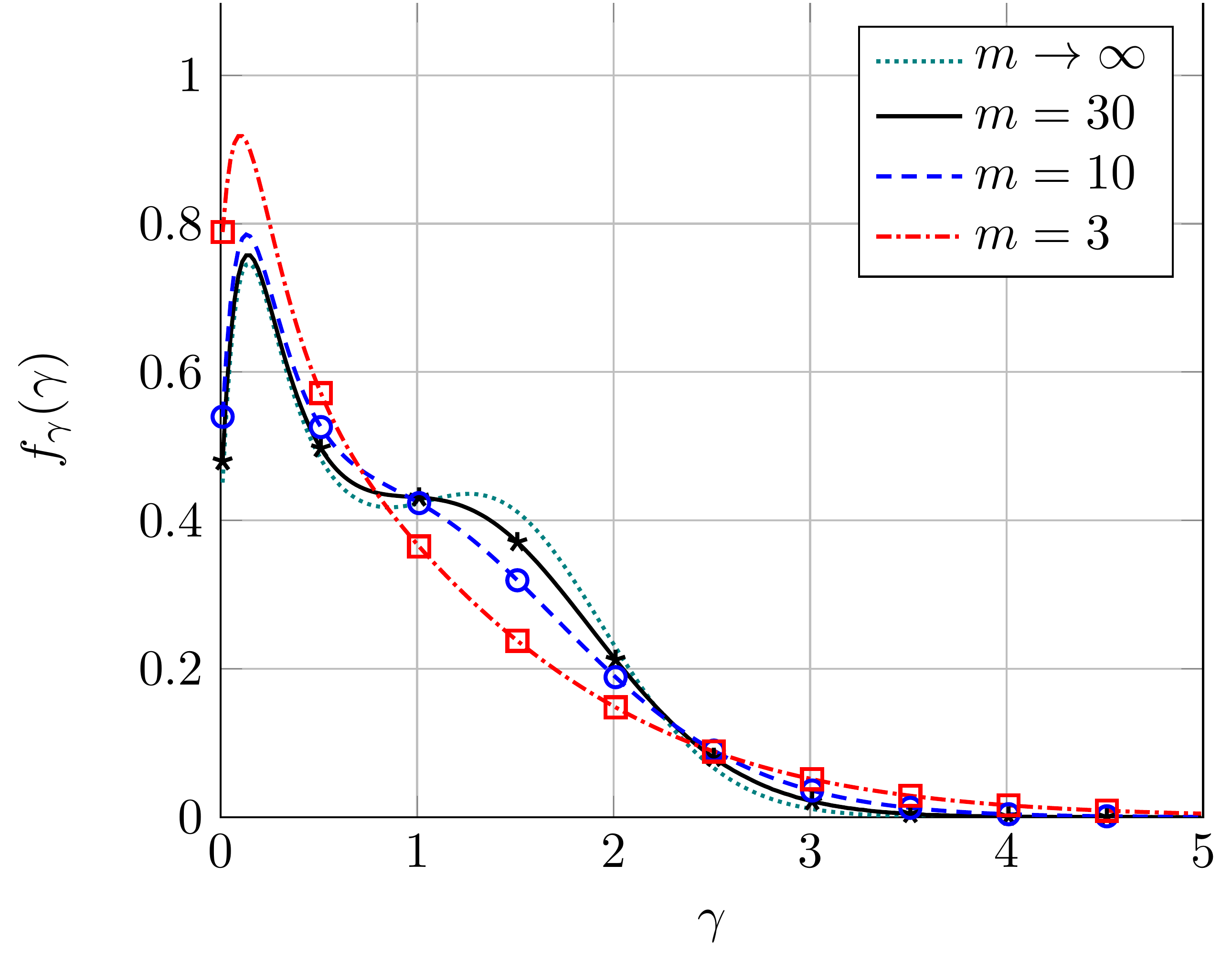}
          \caption{FTR power envelope distribution for different values of $m$, with $K=15$, $\Delta=0.9$ and $\bar\gamma=1$. Solid lines correspond to the exact PDF derived from (\ref{eq:36}), markers correspond to the approximate PDF derived from (\ref{eq:36b}). The case $m\rightarrow\infty$ reduces to the TWDP fading distribution \cite{Durgin2002}.}
    \label{fig:subfig4}
  \end{figure}
}

%

{
\begin{figure}[ht]
    \includegraphics[width=0.99\columnwidth]{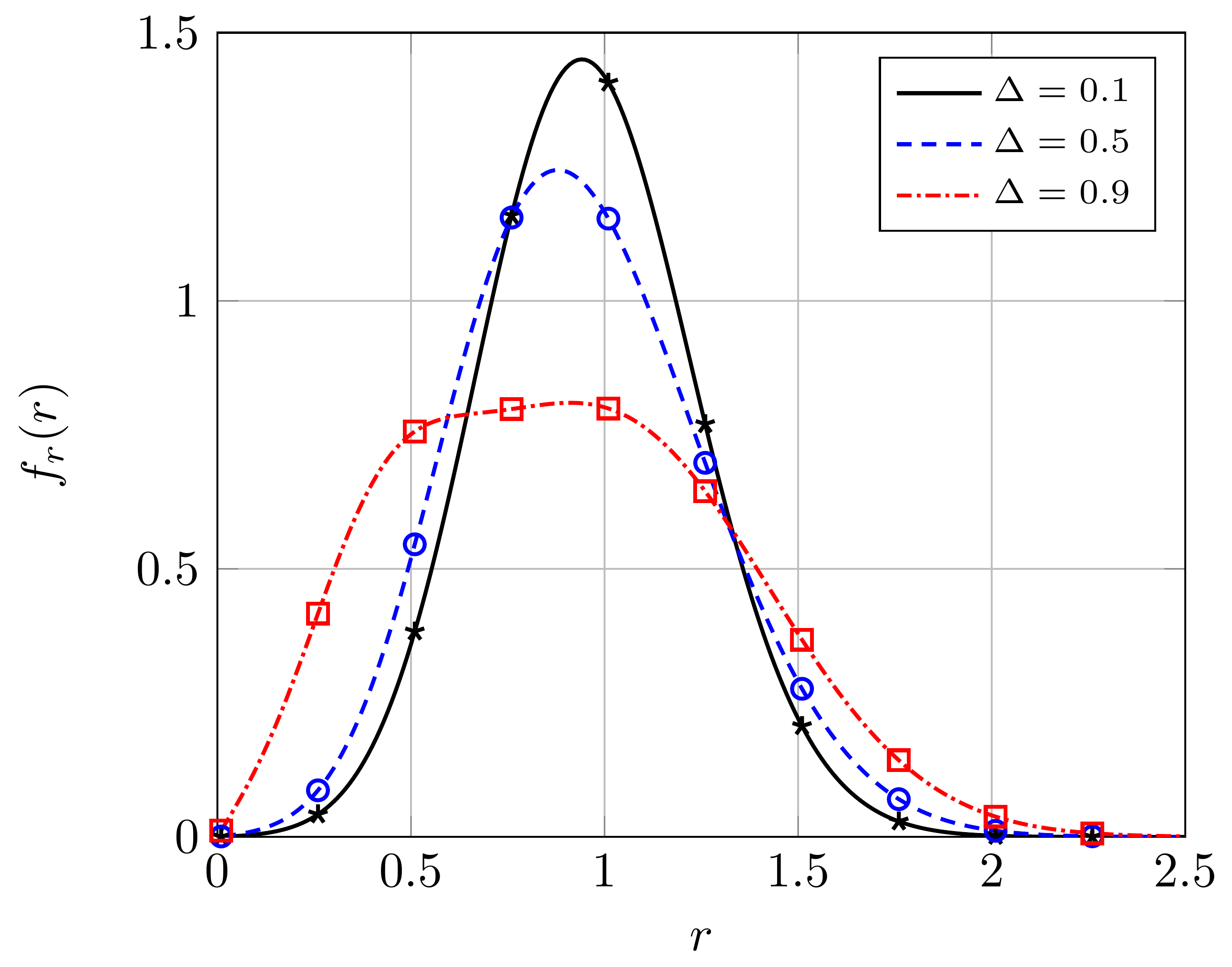}
          \caption{FTR signal envelope distribution for different values of $\Delta$, with $K=15$, $m=5$. Solid lines correspond to the exact PDF derived from (\ref{eq:36}), markers correspond to the approximate PDF derived from (\ref{eq:36b}).}
    \label{fig:subfig2}
  \end{figure}

\begin{figure}[ht]
    \includegraphics[width=0.99\columnwidth]{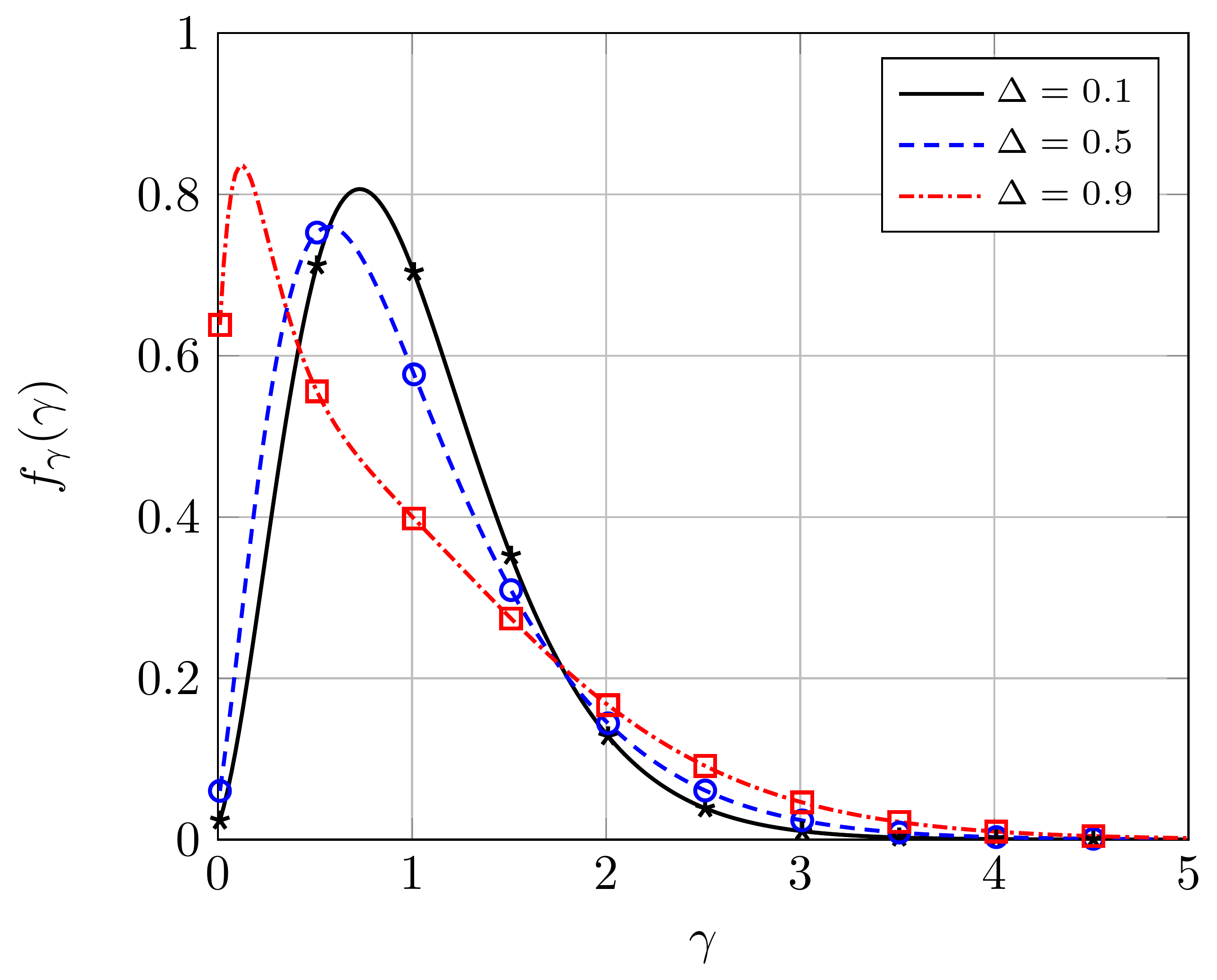}
          \caption{FTR power envelope distribution for different values of $\Delta$, with $K=15$, $m=5$. Solid lines correspond to the exact PDF derived from (\ref{eq:36}), markers correspond to the approximate PDF derived from (\ref{eq:36b}).}
    \label{fig:subfig5}
 \end{figure}
}

\begin{figure}[ht]
    \includegraphics[width=0.99\columnwidth]{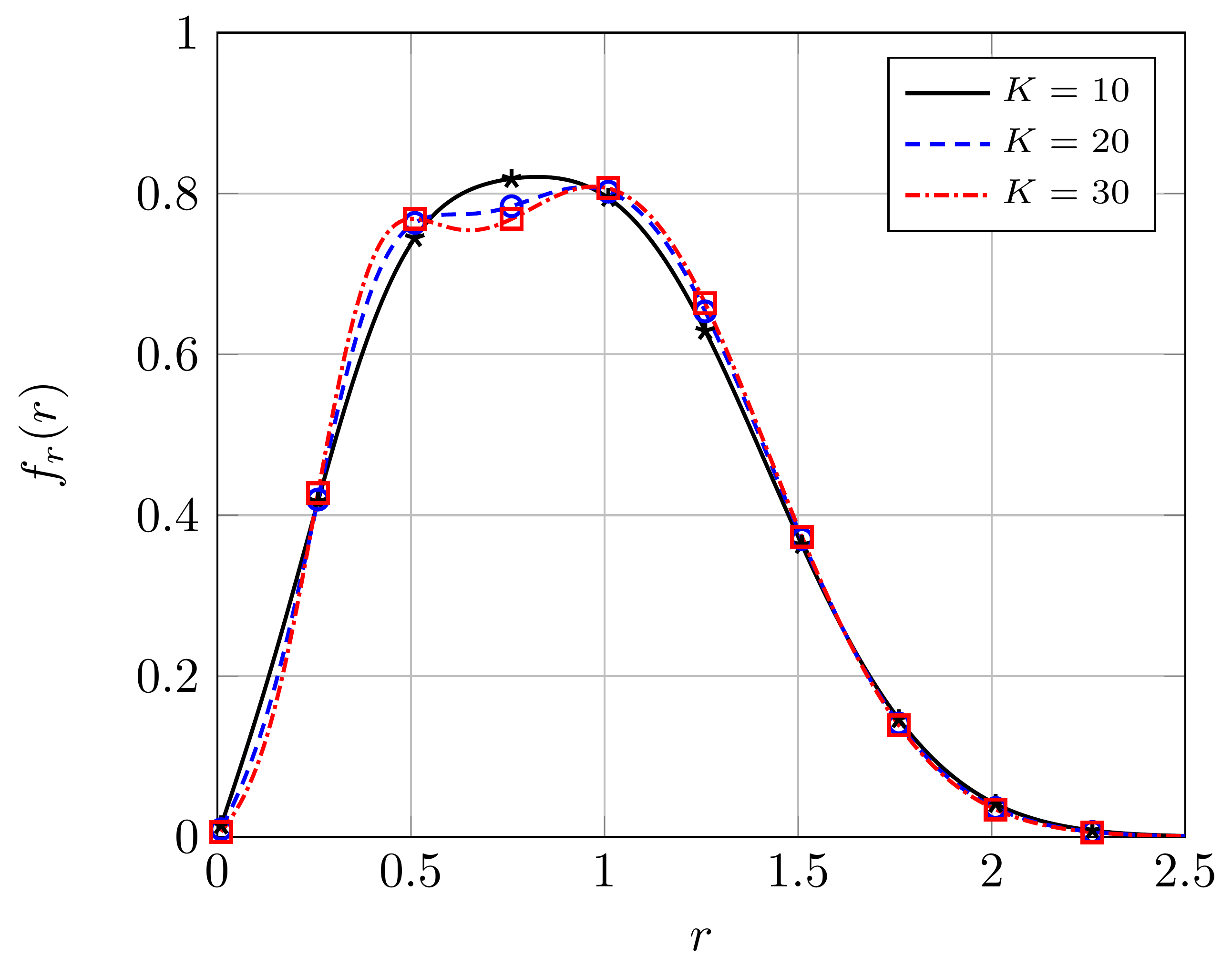}
          \caption{FTR signal envelope distribution for different values of $K$, with $m=5$, $\Delta=0.9$ and $\Omega=1$. Solid lines correspond to the exact PDF derived from (\ref{eq:36}), markers correspond to the approximate PDF derived from (\ref{eq:36b}).}
    \label{fig:subfig3}
  \end{figure}

\begin{figure}[ht]
    \includegraphics[width=0.99\columnwidth]{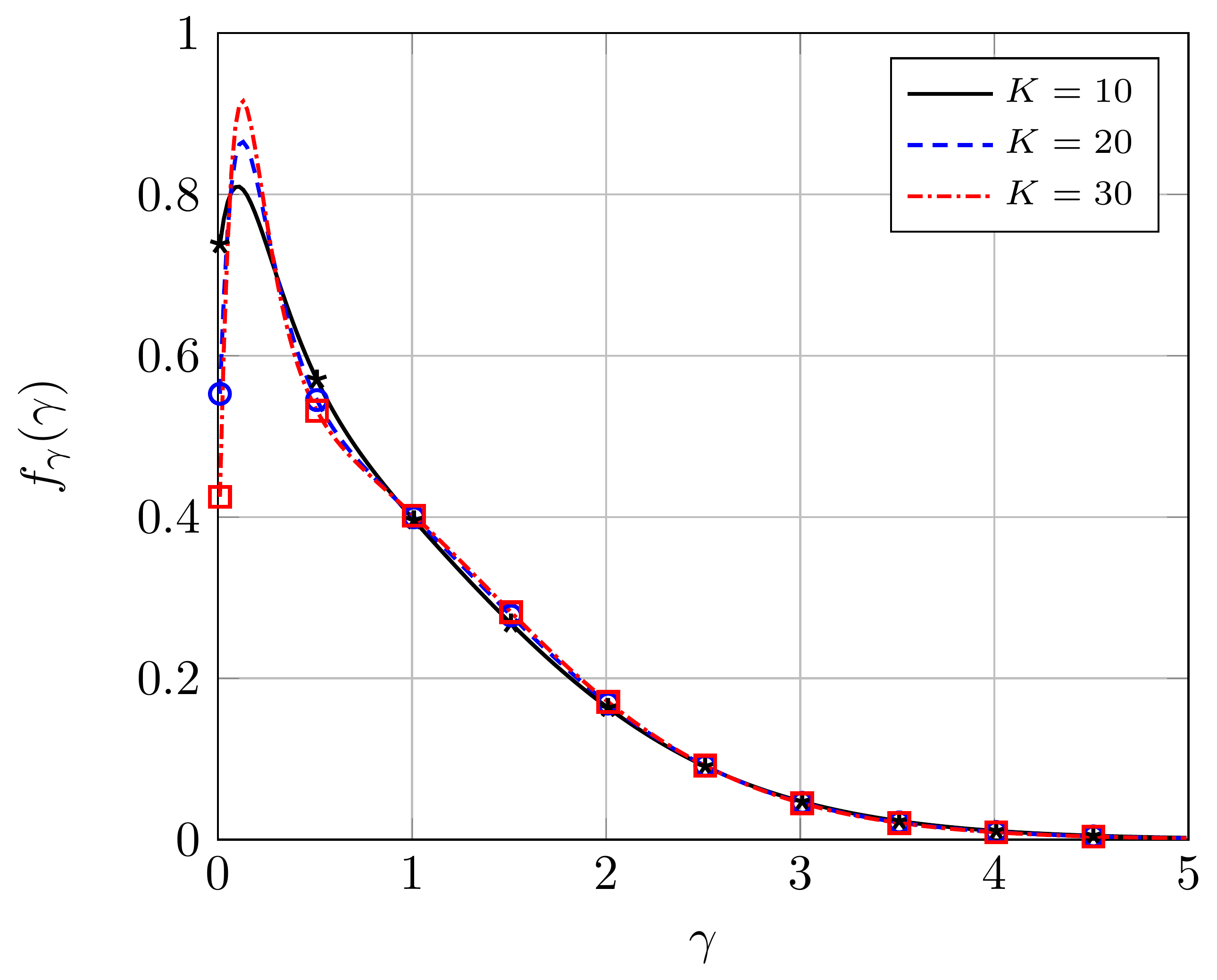}
          \caption{FTR power envelope distribution for different values of $K$, with $m=5$, $\Delta=0.9$ and $\bar\gamma=1$. Solid lines correspond to the exact PDF derived from (\ref{eq:36}), markers correspond to the approximate PDF derived from (\ref{eq:36b}).}
    \label{fig:subfig6}
  \end{figure}


\section{Empirical validation}
\label{fit}

In the previous sections, we have introduced the FTR fading model and derived its relevant statistics. We will now show its suitability for modeling small-scale fading in mmWave wireless links. We use the empirical results presented in \cite{Samimi2016} to validate the FTR fading model in the context of small-scale fading modeling of mmWave outdoor communications in the 28 GHz band. Details on the specific measurement configuration can be found in \cite{Samimi2016}.

A modified version of the Kolmogorov-Smirnov (KS) statistic has been used to define the error factor $\epsilon$ that quantifies the goodness of fit between the empirical and theoretical CDFs, denoted by $\hat{F}_r(\cdot)$ and ${F}_r(\cdot)$ respectively, i.e,
\begin{equation}
\epsilon\triangleq \max_{x}|\log_{10} \hat{F}_r(x)-\log_{10} F_r(x)|.
\end{equation}

Note that the CDF is used in log-scale in order to outweigh the fit in those amplitude values closer to zero, where the fading is more severe \cite{Francis2016}. With the above definition, we must highlight that a value of $\epsilon=1$ can be interpreted as a difference of one order of magnitude between the empirical and theoretical CDFs.

In Figs. \ref{f2} and \ref{f3} we compare the set of measurements corresponding to the LOS and NLOS cross-polarized scenarios described in \cite[Fig. 6]{Samimi2016}. For this set of measurements, the empirical CDFs lie within the theoretical CDFs corresponding to a Rician distribution with values of $K$ ranging from 2 to 7 (i.e. 3 to 8 dB). According to the KS statistic, the values of $K$ that provide the best fit to the Rician distribution are $K_{\text{LOS}}^{\text{Rice}}=4.04$ and $K_{\text{NLOS}}^{\text{Rice}}=4.78$ respectively. Such values of $K$ yield an error factor value of $\epsilon_{\text{LOS}}^{\text{Rice}}=0.3302$ and $\epsilon_{\text{NLOS}}^{\text{Rice}}=0.3571$. Now, using the proposed FTR fading model, we obtain the following set of parameters for the LOS and NLOS cases: ${\text{FTR}_{\text{LOS}}=\left(K=80,\Delta=0.5873,m=2\right)}$ and ${\text{FTR}_{\text{NLOS}}=\left(K=32.7,\Delta=0.8331,m=10\right)}$. 
Note that the parameter $m$ plays a key role in the goodness of fit, as it enables that the CDF can modify its concavity and convexity in order to better adjust the empirical data. 
For these parameters, the error factor value obtained by the FTR fit are $\epsilon_{\text{LOS}}^{\text{FTR}}=0.2246$ and $\epsilon_{\text{NLOS}}^{\text{FTR}}=0.2681$. Thus, a remarkable improvement is attained when using the FTR fading model instead of the simpler Rician model.

\begin{figure}[t]
    \includegraphics[width=1\columnwidth]{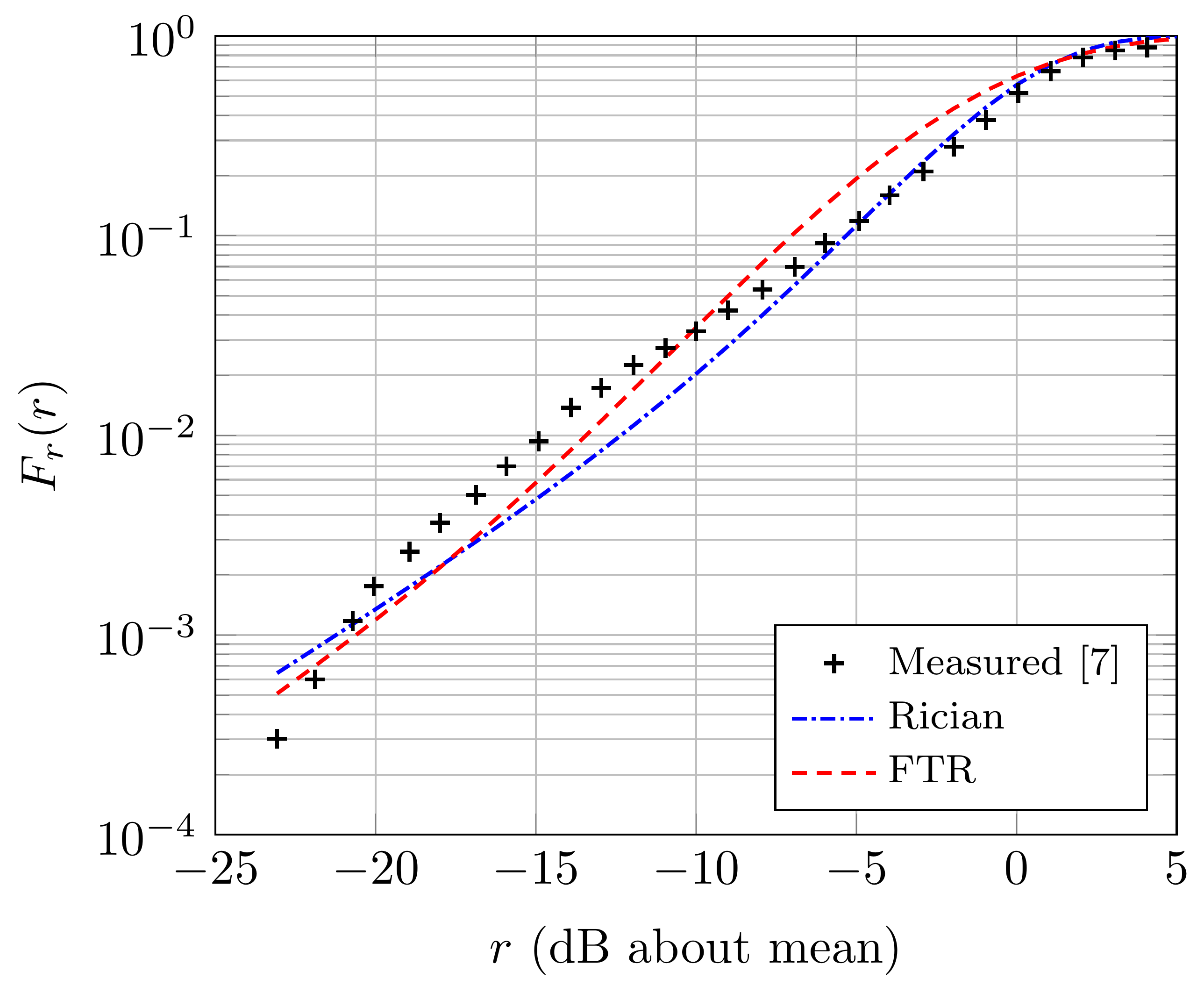}
\caption{Empirical vs theoretical CDFs of the received signal amplitude for LOS scenario. Parameter values are $K_{\text{Rice}}=4.04$ and $K_{\text{FTR}}=80$, $\Delta=0.5873$, $m=2$. Measured data obtained from \cite[Fig. 6, LOS]{Samimi2016}.}
 \label{f2}
 \end{figure}

 \begin{figure}[t]
    \includegraphics[width=1\columnwidth]{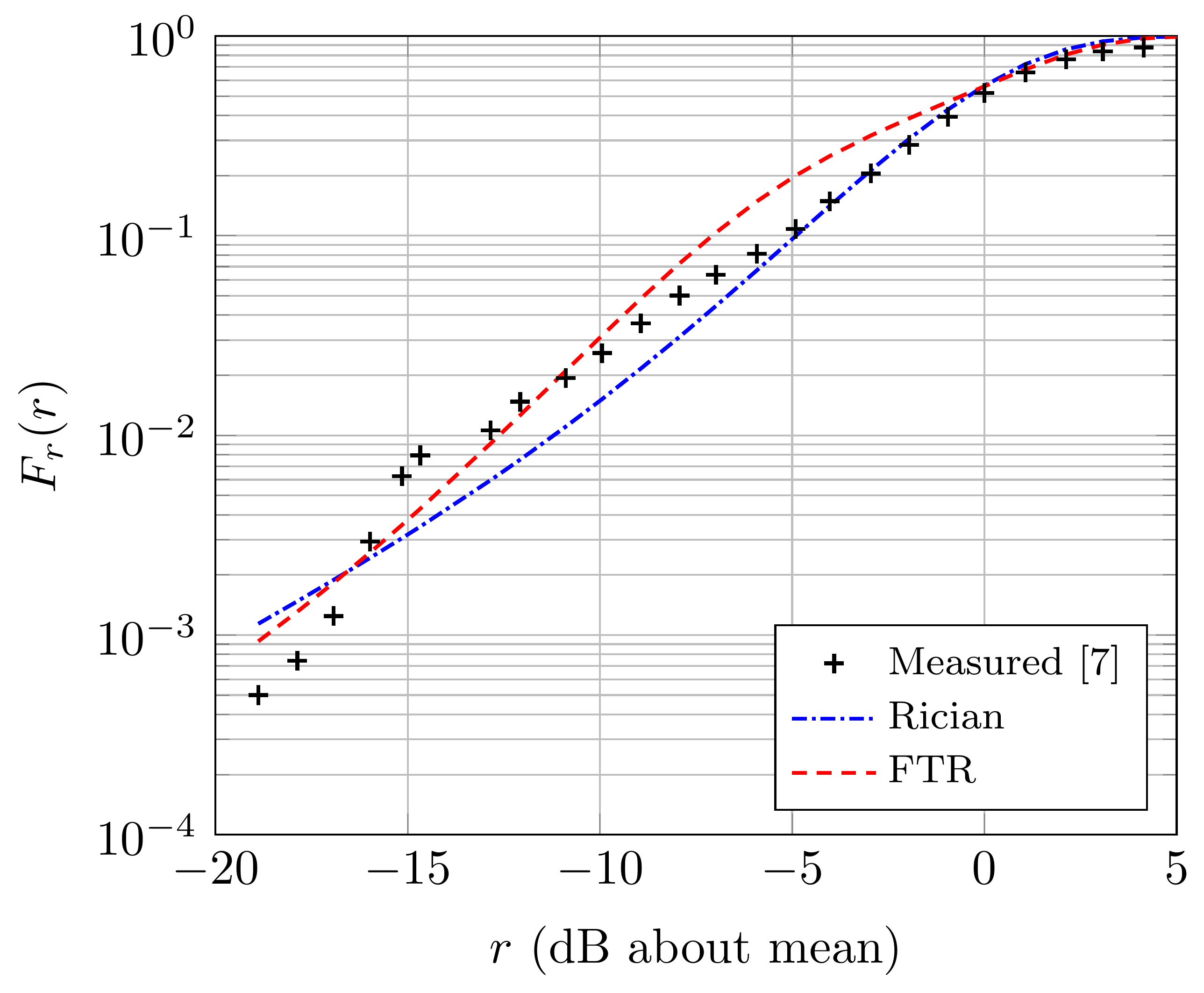}
\caption{Empirical vs theoretical CDFs of the received signal amplitude for NLOS scenario. Parameter values are $K_{\text{Rice}}=4.78$ and ${K_{\text{FTR}}=32.7}$, $\Delta=0.8331$, $m=10$. Measured data obtained from \cite[Fig. 6, NLOS]{Samimi2016}.}
\label{f3} 
\end{figure}
 
\section{Performance analysis of wireless communications systems}
\label{performance}
With the exact closed-form expressions of the MGF, PDF and CDF for the SNR of the proposed FTR fading channel derived above, we can now calculate many performance metrics of wireless communication systems operating in channels following this fading model. As an example of one application, we calculate the BER for a family of coherent modulations and the outage capacity. In both cases we will also obtain exact asymptotic expressions for the high-SNR regime.

\subsection{Average BER }

The average error rates can be calculated by averaging the conditional error probability (CEP), i.e., the error rate under AWGN, over the output SNR, that is:
\begin{equation} \label{eq:41}
	\overline{P_e}=\int^{\infty}_{0} P_E (x) f_\gamma (x) dx,
\end{equation}
where $P_E (x)$ denotes the CEP. Alternatively, integrating (\ref{eq:41}) by parts, the average error rate can be computed from the CDF as
\begin{equation} \label{eq:42}
	\overline{P_e}=-\int^{\infty}_{0} P^{'}_E (x) F_\gamma (x) dx,
\end{equation}
where $P^{'}_E(\gamma)$ is the first order derivative of the CEP.

The CEP for the bit error rate of many wireless communication systems with coherent detection is determined by \cite{Lopez2010}
\begin{equation}
\label{eq:43}
P_E \left( x \right) = \sum\limits_{r = 1}^R {\alpha _r Q\left( {\sqrt {\beta _r x} } \right)}, 
\end{equation}
where $Q(\cdot)$ is the Gauss $Q$-function and $\left\{ {\alpha _r ,\beta _r } \right\}_{r = 1}^R $ are modulation-dependent constants. The derivative of (\ref{eq:43}) is given by\begin{equation}
\label{eq:44}
P_E^{'}  \left( x \right) =  - \sum\limits_{r = 1}^R {\alpha _r \sqrt {\frac{{\beta _r }}
{{8\pi x}}} } e^{ - \frac{{\beta _r x}}
{2}}. 
\end{equation}

Introducing (\ref{eq:37}) and (\ref{eq:44}) into (\ref{eq:42}), and with the help of \cite[p. 286, (43)]{Srivastava1985}, a compact exact expression of the average BER can be found, as given in (\ref{eq:45}), in terms of the Lauricella function $F_D(\cdot)$ defined in \cite[p. 33, (4)]{Srivastava1985}.

Although the derived BER expression can be easily computed using the Euler form of the $F_D$ function, it does not provide insight about the impact of the different system parameters on performance. We now present an asymptotic, yet accurate, simple expression of the error rates for the  high SNR regime. First, note that the following equality holds:
\begin{equation}
\label{eq:46}
\begin{split}
  &  \left| {M_\gamma  \left( s \right)} \right| = \frac{{m^m \left( {1 + K} \right)}}
{{\left( {\sqrt {\left( {m + K} \right)^2  - \Delta ^2 K^2 } } \right)^m }}  \cr 
  &  \times P_{m - 1} \left( {\frac{{\left( {m + K} \right)}}
{{\sqrt {\left( {m + K} \right)^2  - \Delta ^2 K^2 } }}} \right)\frac{1}
{{\bar \gamma \left| s \right|}} + o\left( {\left| s \right|^{ - 1} } \right), \cr
\end{split}
\end{equation}
where we write a function $a(x)$  as $o(x)$ if $\lim_{x\rightarrow\infty} a(x)/x = 0$ and where the Legendre polynomial is calculated using  (\ref{eq:34}). Thus, performing a similar approach to that in \cite[Propositions 1 and 3]{Wang03}, we obtain, after some manipulation, the asymptotic expression
\begin{equation}
\label{eq:47}
\begin{split}
  & \overline{P_e} \approx \frac{{m^m \left( {1 + K} \right)}}
{{2\left( {\sqrt {\left( {m + K} \right)^2  - \Delta ^2 K^2 } } \right)^m }}\left( {\sum\limits_{r = 1}^R {\frac{{\alpha _r }}
{{\beta _r }}} } \right)  \cr 
  & \quad \quad  \times P_{m - 1} \left( {\frac{{\left( {m + K} \right)}}
{{\sqrt {\left( {m + K} \right)^2  - \Delta ^2 K^2 } }}} \right)\frac{1}
{{\bar \gamma }},\quad \bar \gamma  \gg 1. \cr
\end{split}
\end{equation}

\begin{figure*}[!t]
\normalsize
\begin{equation}
\label{eq:45}
\begin{split}
& \overline{P_e} = \frac{{ 1}}
{{2^{m - 1} }}\frac{{1 + K}}
{{\bar \gamma }}\left( {\frac{m}
{{\sqrt {\left( {m + K} \right)^2  - K^2 \Delta ^2 } }}} \right)^m \sum\limits_{q = 0}^{\left\lfloor {(m - 1)/2} \right\rfloor } {\left( { - 1} \right)^q } C_q^{m - 1} \left( {\frac{{m + K}}
{{\sqrt {\left( {m + K} \right)^2  - K^2 \Delta ^2 } }}} \right)^{m - 1 - 2q} \sum\limits_{r = 1}^R {\alpha _r \frac{1}
{{2\beta _r }}}   \cr 
  & \quad  \times \;F_D^{(4)} \left( {\frac{3}
{2},1 + 2q - m,m - q - \frac{1}
{2},m - q - \frac{1}
{2},1 - m;2;} \right.  \cr 
  & \quad \quad \quad \quad \quad \quad \quad \quad \left. { - \frac{{2m\left( {1 + K} \right)}}
{{\beta _r \left( {m + K} \right)\bar \gamma }}, - \frac{{2m\left( {1 + K} \right)}}
{{\beta _r \left( {m + K\left( {1 + \Delta } \right)} \right)\bar \gamma }}, - \frac{{2m\left( {1 + K} \right)}}
{{\beta _r \left( {m + K\left( {1 - \Delta } \right)} \right)\bar \gamma }}, - \frac{{2\left( {1 + K} \right)}}
{{\beta _r \bar \gamma }}} \right). \cr
\end{split}
\end{equation}
\hrulefill
\vspace*{4pt}
\end{figure*}

\subsection{Outage capacity }

The instantaneous channel capacity per unit bandwidth considering transmit and receive antennas is well-known to be given by
\begin{equation}
\label{eq:48}
C = \log _2 (1 + \gamma ).
\end{equation}
We define the outage capacity probability as the probability that the instantaneous channel capacity $C$ falls below a predefined threshold  $R_S$ (given in terms of rate per unit bandwidth), i.e.,
\begin{equation}
\label{eq:49}
P_{out}  = P\left( {C < R_S } \right) = P\left( {\log _2 (1 + \gamma ) < R_S } \right).
\end{equation}
Therefore
\begin{equation}
\label{eq:50}
P_{out}  = P\left( {\gamma  < 2^{R_S }  - 1} \right) = F_\gamma  \left( {2^{R_S }  - 1} \right).
\end{equation}
Thus, the outage capacity probability can be directly calculated from (\ref{eq:37}) specialized for $x=2^{R_S }  - 1$. This expression is exact and holds for all SNR values; however, it offers little insight about the effect of parameters on performance. Fortunately, we can obtain a simple expression in the high SNR regime as follows: From (\ref{eq:37}) and \cite[Proposition 5]{Wang03}, the CDF of $\gamma$ can be written as
\begin{equation}
\label{eq:51}
\begin{split}
  & F_\gamma  (x) = \frac{{m^m \left( {1 + K} \right)}}
{{\left( {\sqrt {\left( {m + K} \right)^2  - \Delta ^2 K^2 } } \right)^m }}  \cr 
  & \quad \quad P_{m - 1} \left( {\frac{{\left( {m + K} \right)}}
{{\sqrt {\left( {m + K} \right)^2  - \Delta ^2 K^2 } }}} \right)\frac{x}
{{\bar \gamma }} + o\left( {\bar \gamma ^{ - 1} } \right). \cr 
\end{split}
\end{equation}
Therefore, the outage capacity probability can be approximated in the large SNR regime by
\begin{equation}
\label{eq:52}
\begin{split}
  & P_{out}  \approx \frac{{m^m \left( {1 + K} \right)}}
{{\left( {\sqrt {\left( {m + K} \right)^2  - \Delta ^2 K^2 } } \right)^m }}  \cr 
  & \quad \quad P_{m - 1} \left( {\frac{{\left( {m + K} \right)}}
{{\sqrt {\left( {m + K} \right)^2  - \Delta ^2 K^2 } }}} \right)\frac{{2^{R_S }  - 1}}
{{\bar \gamma }},\quad \bar \gamma  \gg 1. \cr
\end{split}
\end{equation}


\section{Numerical results}
\label{numerical}

In this section we present figures showing numerical results for the performance metrics derived in the previous section under different fading conditions. All the results shown here have been analytically obtained by the direct evaluation of the exact expressions derived in this paper. Additionally, Monte Carlo simulations have been performed to validate the obtained expressions, and are also presented in the figures, showing an excellent agreement with the analytical results. 

Figs. \ref{fig:subfig7} and \ref{fig:subfig8} show results for, respectively, the average BER and the outage capacity probability as a function of the average SNR assuming $K=8$ and similar ($\Delta=0.9$) and dissimilar ($\Delta=0.1$) specular components, as well as light ($m=8$) and  strong ($m=2$) fluctuations of these components. For the average BER, binary phase-shift keying (BPSK) modulation is assumed, which can be obtained by setting $R=1$, $\alpha_1=1$ and $\beta_1=2$ in (\ref{eq:45}). For the outage capacity probability, a threshold $R_S=2$ is assumed. Both presented metrics actually show akin behavior: dissimilar specular components experiencing lighter fluctuations yield better performance, i.e. lower average BER and outage capacity probability. It is interesting to note that, for both metrics, there is an inflection point (in the log-log scale)  for ($\Delta=0.1$) and ($m=8$), which virtually disappears as $m$ decreases.

\begin{figure}[t]
     \includegraphics[width=1\columnwidth]{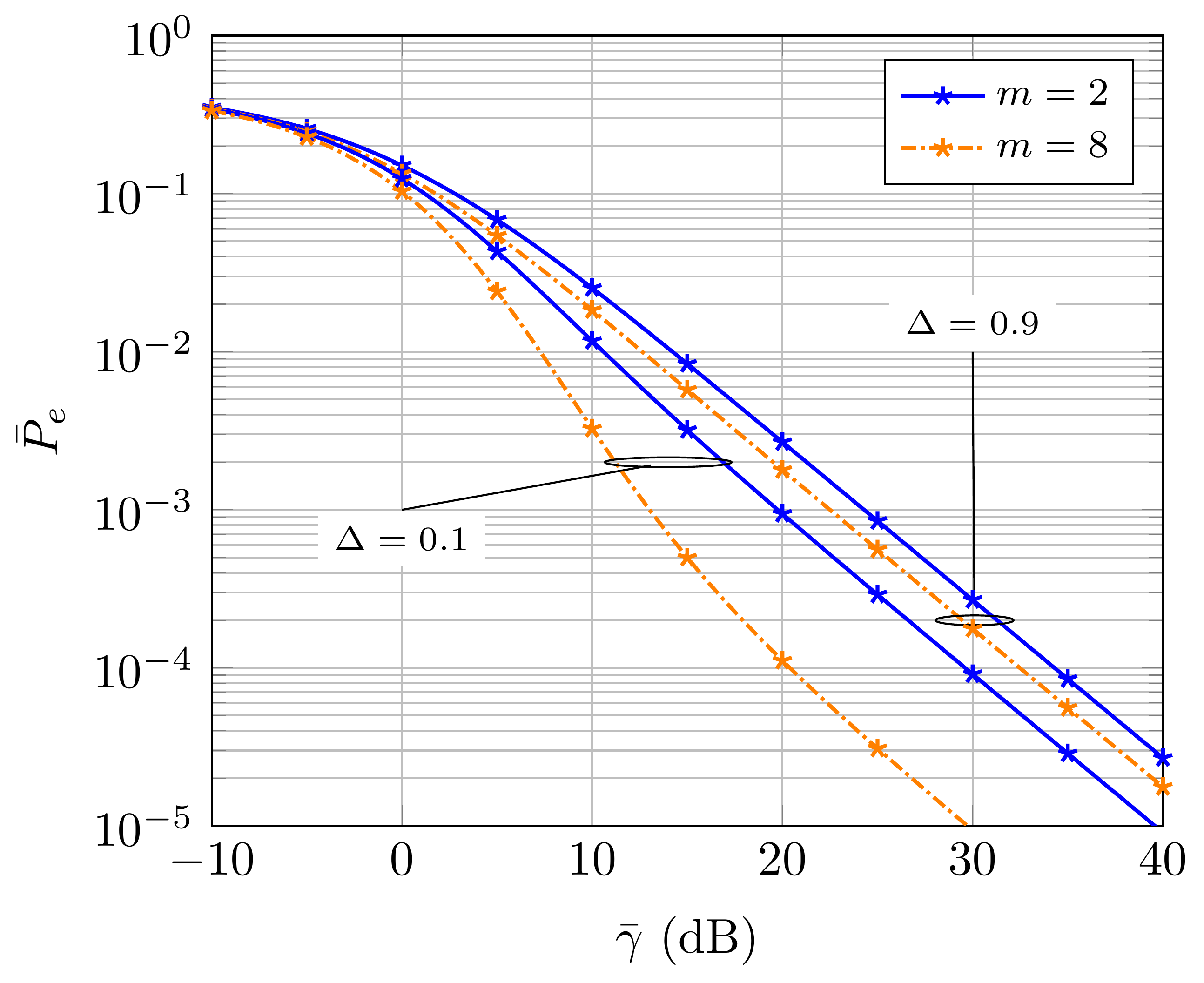}
          \caption{Average BER of BPSK modulation vs. average SNR for different values of $m$ and $\Delta$. Parameter value $K=8$. Markers correspond to Monte Carlo simulations.}
    \label{fig:subfig7}
 \end{figure}

\begin{figure}[t]
    \includegraphics[width=1\columnwidth]{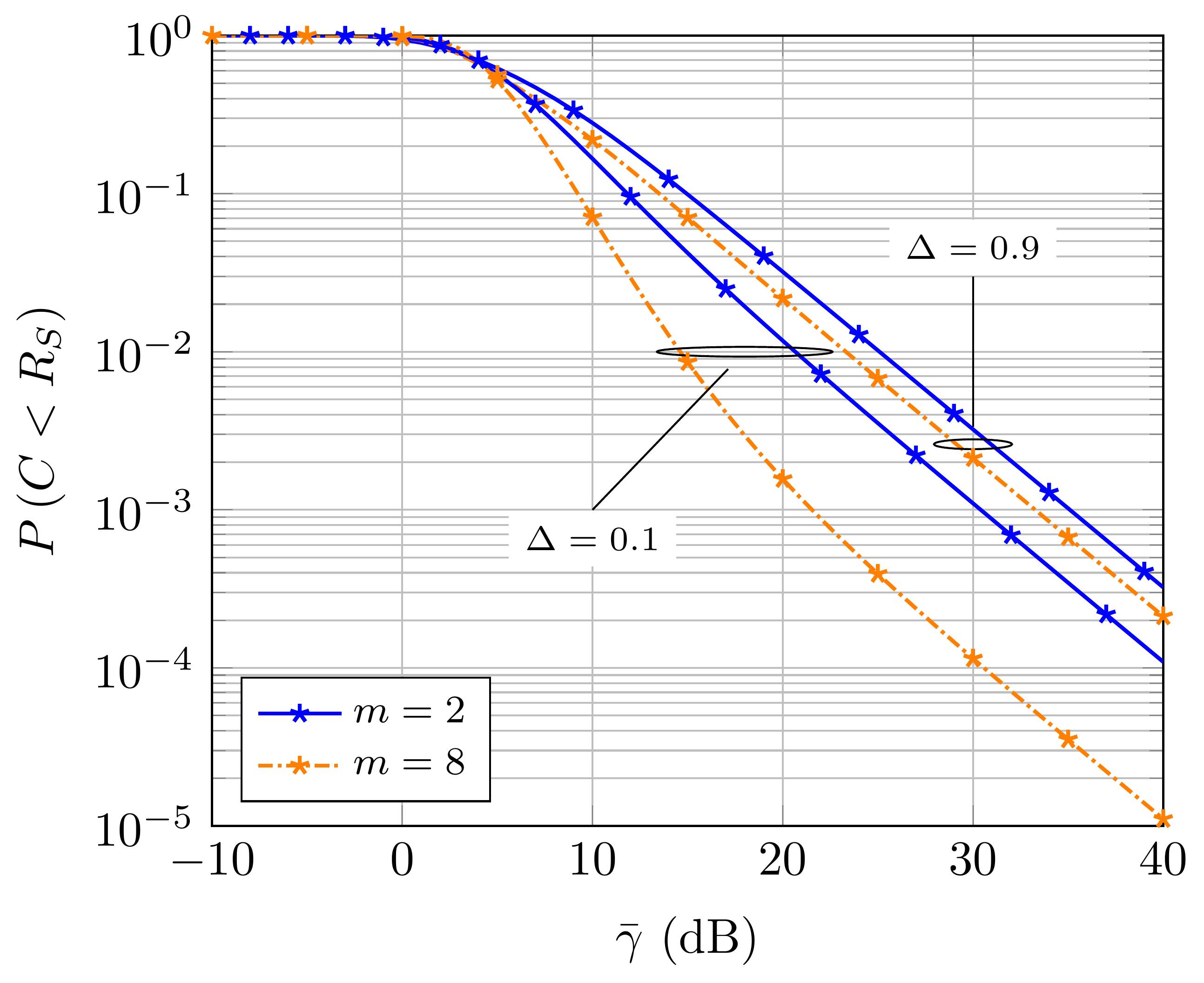}
          \caption{Outage capacity probability vs. average SNR for different values of $m$ and $\Delta$. Parameter values $R_S=2$ and $K=8$. Markers correspond to Monte Carlo simulations.}
    \label{fig:subfig8}
 \end{figure}

\section{Further generalizations of the TWDP fading model}
The FTR model, as presented and analyzed in this work, assumes that the two specular components are fully correlated and, therefore, is an appropriate model when both specular components are affected by the same scatterers or electromagnetic disturbances. A more general model would consider a partial correlation between the specular components. The analytical difficulty of such model seems to be significantly higher than the one presented here. However, in the limit case when both specular components are independent and the complex baseband received signal can be expressed as
\begin{equation}
\label{eq:53}
\begin{split}
V_r  = V_1 \sqrt {\zeta _1 } \exp \left( {j\phi _1 } \right) + V_2 \sqrt {\zeta _2 } \exp \left( {j\phi _2 } \right) + X + jY,
\end{split}
\end{equation}
where $\zeta_1$ and $\zeta_2$ are i.i.d. random variables whose PDF is given in (\ref{eq:03}), it is possible to obtain the MGF of the received SNR in closed-form.
\begin{lemma}  
Let us consider the fading model as described in (\ref{eq:53}). Then, the MGF of the received SNR $\gamma$ will be given by
\begin{equation}
\label{eq:54}
\begin{split}
  & M_\gamma  \left( s \right) = {{1 + K} \over {1 + K - \bar \gamma s}}\left( {1 - {K \over m}g(s) + {{K^2 \Delta ^2 } \over {4m^2 }}g^2 (s)} \right)^{ - m}   \cr 
  & \quad \quad  \times _2 F_1 \left( {m ,m ;1;{{K^2 \Delta ^2 g^2 (s)} \over {4m^2  - 4mKg(s) + K^2 \Delta ^2 g^2 (s)}}} \right), 
\end{split}
\end{equation}
where $K$ and $\Delta$ are defined, respectively, as in (\ref{eq:05}) and (\ref{eq:06}), and $g(s)$ is defined as
\begin{equation}
\label{eq:55}
\begin{split}
g(s) = {{\bar \gamma s} \over {1 + K - \bar \gamma s}}.
\end{split}
\end{equation}
\end{lemma} 
\begin{proof} See Appendix \ref{App5}.
\end{proof}

This model is an alternative generalization of the TWDP model which is different to the proposed FTR model, and can be applied when the specular components follow very different paths and are affected by different scatterers. A thorough analysis of such model is of great interest and is left for future work.

\section{Conclusions}
\label{conc}

The FTR fading model was introduced to characterize the statistics of a received signal with dominant specular components along with random fluctuations about those components. A detailed statistical characterization is presented, and closed-form expressions of the PDF, CDF and MGF of the model are derived. Additionally, as the exact PDF and CDF are given in terms of a confluent hypergeometric function, alternative approximated expressions for these statistics are given as finite summations  of elementary functions, which allows for a simple performance analysis. As an example of application of the model, we have derived both exact and asymptotic expressions for the outage capacity and the BER for a family of modulation schemes. Both performance metrics show that dissimilar specular components experiencing lighter fluctuations yield better performance. The proposed model is also shown to closely model small-scale fading, which has been exemplified in the context of mmWave communications, on which the fit to empirical measurements in the 28 GHz band shows great improvements over the Rician fading model. 

\appendices

\section{Proof of Lemma 1}
\label{App1}

Let us consider the fading channel model given in (\ref{eq:04}) conditioned to a particular realization $\zeta=u$ of the random variable modeling the fluctuation of the specular components. In this case, we can write
\begin{equation}
\label{eq:12}
\left. {V_r } \right|_{\zeta  = u}  = \sqrt u V_1 \exp \left( {j\phi _1 } \right) + \sqrt u V_2 \exp \left( {j\phi _2 } \right) + X + jY
\end{equation}
This corresponds to the classical TWDP fading model where the amplitudes of the specular components are given by $\sqrt u V_1$ and $\sqrt u V_2$, for which the following parameters can be defined:
\begin{equation}
\label{eq:13}
K_u  = \frac{{uV_1^2  + uV_2^2 }}
{{2\sigma ^2 }} = u\frac{{V_1^2  + V_2^2 }}
{{2\sigma ^2 }},
\end{equation}
\begin{equation}
\label{eq:14}
\Delta _u  = \frac{{2\sqrt u V_1 \sqrt u  V_2 }}
{{uV_1^2  + uV_2^2 }} = \frac{{2V_1 V_2 }}
{{V_1^2  + V_2^2 }}.
\end{equation}
It is clear that these parameters are related to those defined in (\ref{eq:05}) and (\ref{eq:06}) for the FTR fading model by
\begin{equation}
\label{eq:15}
K_u  = uK,
\end{equation}
\begin{equation}
\label{eq:16}
\Delta_u = \Delta.
\end{equation}
The conditional average SNR for the fading model described in (\ref{eq:12}) will be
\begin{equation}
\label{eq:17}
\begin{split}
  & \bar \gamma _u  = \left( {E_b /N_0 } \right)\left( {uV_1^2  + uV_2^2  + 2\sigma ^2 } \right)  \cr 
  & \quad \, = \left( {E_b /N_0 } \right)2\sigma ^2 \left( {1 + K_u } \right). \cr
\end{split}
\end{equation}
The MGF of the TWDP fading model was shown in \cite {Rao2015} to be given in closed-form as
\begin{equation}
\label{eq:18}
\begin{split}
  & M_{\gamma _u } \left( s \right) = \frac{{1 + K_u }}
{{1 + K_u  - \bar \gamma _u s}}\exp \left( {\frac{{K_u \bar \gamma _u s}}
{{1 + K_u  - \bar \gamma _u s}}} \right)  \cr 
  & \quad \quad \quad \quad \quad  \times I_0 \left( {\frac{{\Delta _u K_u \bar \gamma _u s}}
{{1 + K_u  - \bar \gamma _u s}}} \right), \cr
\end{split}
\end{equation}
where $I_0(\cdot)$ is the zero-order modified Bessel function of the first kind. This MGF can be written in terms of the $K$ and $\Delta$ parameters defined for the FTR model. Note that from (\ref{eq:07}) and (\ref{eq:17}) we can write, respectively, 
\begin{equation}
\label{eq:19}
\frac{{\left( {1 + K} \right)}}
{{\bar \gamma }} = \frac{1}
{{\left( {E_b /N_0 } \right)2\sigma ^2 }},
\end{equation}
\begin{equation}
\label{eq:20}
\frac{{\left( {1 + K_u } \right)}}
{{\bar \gamma _u }} = \frac{1}
{{\left( {E_b /N_0 } \right)2\sigma ^2 }},
\end{equation}
and equating (\ref{eq:19}) and (\ref{eq:20}) it is clear that
\begin{equation}
\label{eq:21}
\frac{{\left( {1 + K_u } \right)}}
{{\bar \gamma _u }} = \frac{{\left( {1 + K} \right)}}
{{\bar \gamma }}.
\end{equation}
Now, taking into account (\ref{eq:15}), (\ref{eq:16}) and (\ref{eq:21}), we have
\begin{equation}
\label{eq:22}
\begin{split}
  & M_{\gamma _u } \left( s \right) = \frac{{\frac{{1 + K_u }}
{{\bar \gamma _u }}}}
{{\frac{{1 + K_u }}
{{\bar \gamma _u }} - s}}\exp \left( {\frac{{K_u s}}
{{\frac{{1 + K_u }}
{{\bar \gamma _u }} - s}}} \right)  \cr 
  & \quad \quad \quad \quad \quad \quad \times I_0 \left( {\Delta _u \frac{{K_u s}}
{{\frac{{1 + K_u }}
{{\bar \gamma _u }} - s}}} \right)  \cr 
  & \quad   = \frac{{\frac{{1 + K}}
{{\bar \gamma }}}}
{{\frac{{1 + K}}
{{\bar \gamma }} - s}}\exp \left( {\frac{{uKs}}
{{\frac{{1 + K}}
{{\bar \gamma }} - s}}} \right)   I_0 \left( {\Delta \frac{{uKs}}
{{\frac{{1 + K}}
{{\bar \gamma }} - s}}} \right), \cr
\end{split}
\end{equation}
and therefore the conditional MGF can be written as
\begin{equation}
\label{eq:23}
M_{\gamma _u } \left( s \right) = \mathcal{B}\left( s \right)e^{u \mathcal{A}\left( s \right)} I_0 \left( {u\Delta \mathcal{A}\left( s \right)} \right),
\end{equation}
where we have defined
\begin{equation}
\label{eq:24}
\mathcal{A}\left( s \right) = \frac{{K\bar \gamma s}}
{{1 + K - \bar \gamma s}},
\quad \quad
\mathcal{B}\left( s \right) = \frac{{1 + K}}
{{1 + K - \bar \gamma s}}.
\end{equation}

The MGF of the SNR of the FTR model can be obtained by averaging (\ref{eq:23}) over all possible realizations $u$ of the random variable $\zeta$, i.e., 
\begin{equation}
\label{eq:26}
\begin{split}
  & M_\gamma  \left( s \right) = \int_0^\infty  {M_{\gamma _u } \left( s \right)} f_\zeta  \left( u \right)du   \cr 
  & = B\left( s \right)\frac{{m^m }}
{{\Gamma \left( m \right)}}\int_0^\infty  {u^{m - 1} e^{ - u\left( {m - A\left( s \right)} \right)} I_0 \left( {u\Delta A\left( s \right)} \right)} du. \cr
\end{split}
\end{equation}
The integral in  (\ref{eq:26}) can be solved in closed-form, as from \cite[p. 196 (8)]{Erdelyi1854} we have
\begin{equation}
\label{eq:27}
\int_0^\infty  {t^\mu  e^{ - \beta t} I_0 \left( {\alpha t} \right)} dt = \Gamma \left( {\mu  + 1} \right)\theta ^{ - \mu  - 1} P_\mu  \left( {\beta /\theta } \right),
\end{equation}
where $\theta  = \sqrt {\beta ^2  - \alpha ^2 } $. Using  (\ref{eq:27}) to solve  (\ref{eq:26}), after some algebraic manipulations,  (\ref{eq:08}) is obtained.

\section{Proof of corollary 1}
\label{App2}

For $m=1$, the Legendre function in the MGF given in (\ref{eq:08}) has a zero degree. Taking into account that $P_0(z)=1$ for all $z$, we can write
\begin{equation}
\label{eq:29}
M_\gamma  \left( s \right) = \frac{{\left( {1 + K} \right)}}
{{\sqrt {\mathcal{R}\left( {1,K,\Delta ;s} \right)} }}.
\end{equation}
Considering now that
\begin{equation}
\label{eq:30}
\begin{split}
  & R\left( {1,K,\Delta ;s} \right)) = \left[ {\left( {1 + K} \right)^2  - \Delta ^2 K^2 } \right]\bar \gamma ^2 s^2   \cr 
  & \quad  - 2\left( {1 + K} \right)^2 \bar \gamma s + \left( {1 + K} \right)^2   \cr 
  & \quad  = \left( {1 + K} \right)^2 \left[ {\left( {1 - \frac{{\Delta ^2 K^2 }}
{{\left( {1 + K} \right)^2 }}} \right)\bar \gamma ^2 s^2  - 2\bar \gamma s + 1} \right], \cr
\end{split}
\end{equation}
and introducing (\ref{eq:30}) into (\ref{eq:29}) we obtain
\begin{equation}
\label{eq:31}
M_\gamma  \left( s \right) = \frac{1}
{{\sqrt {\left( {1 - \frac{{\Delta ^2 K^2 }}
{{\left( {1 + K} \right)^2 }}} \right)\bar \gamma ^2 s^2  - 2\bar \gamma s + 1} }}.
\end{equation}
By noting that the MGF of the SNR in Nakagami-$q$ (Hoyt) fading is given by
\begin{equation}
\label{eq:32}
M_{Hoyt} \left( s \right)  = \frac{1}
{{\sqrt {\frac{{4q^2 }}
{{\left( {1 + q^2 } \right)^2 }}\bar \gamma ^2 s^2  - 2\bar \gamma s + 1} }},
\end{equation}
and equating
\begin{equation}
\label{eq:33}
1 - \frac{{\Delta ^2 K^2 }}
{{\left( {1 + K} \right)^2 }} = \frac{{4q^2 }}
{{\left( {1 + q^2 } \right)^2 }},
\end{equation}
the expression given in (\ref{eq:28})  for the $q$ parameter is finally obtained.

\section{Proof of Lemma 2}
\label{App3}

We note that the polynomial $\mathcal{R}\left( {m,k,\Delta ;s} \right)$ defined in (\ref{eq:09}) can be factorized as
\begin{equation}
\label{eq:38}
\begin{split}
  & \mathcal{R}\left( {m,K,\Delta ;s} \right) = \left[ {m\left( {1 + K} \right) - \left( {m + K\left( {1 + \Delta } \right)} \right)\bar \gamma s} \right]  \cr 
  & \quad \quad \quad  \times \left[ {m\left( {1 + K} \right) - \left( {m + K\left( {1 - \Delta } \right)} \right)\bar \gamma s} \right]. \cr
\end{split}
\end{equation}
For the sake of compactness, let us define the following parameters;
\begin{equation}
\label{eq:39}
\begin{split}
  & a_1  = \frac{{m\left( {1 + K} \right)}}
{{\left( {m + K} \right)\bar \gamma }},\quad \quad \quad a_2  = \frac{{m\left( {1 + K} \right)}}
{{\left( {m + K\left( {1 + \Delta } \right)} \right)\bar \gamma }},\quad   \cr 
  & a_3  = \frac{{m\left( {1 + K} \right)}}
{{\left( {m + K\left( {1 - \Delta } \right)} \right)\bar \gamma }},\quad a_4  = \frac{{1 + K}}
{{\bar \gamma }}. \cr
\end{split}
\end{equation}
From (\ref{eq:08}),  using (\ref{eq:34}) and (\ref{eq:38}), the MGF of $\gamma$ can be rewritten as
\begin{equation}
\label{eq:40}
\begin{split}
  & M_\gamma  \left( s \right) = \frac{{ - \left( {a_2 a_3 } \right)^{\frac{m}
{2}} }}
{{\left( {2a_4 } \right)^{m - 1} }}\sum\limits_{q = 0}^{\left\lfloor {(m - 1)/2} \right\rfloor } {\left( { - 1} \right)^q }   \cr 
  & \; \times C_q^{m - 1} \left( {\frac{{\left( {a_2 a_3 } \right)^{\frac{1}
{2}} }}
{{a_1 }}} \right)^{m - 1 - 2q} \frac{1}
{s}\left( {1 - \frac{{a_1 }}
{s}} \right)^{m - 1 - 2q}   \cr 
  & \; \times \left( {1 - \frac{{a_2 }}
{s}} \right)^{\frac{1}
{2} + q - m} \left( {1 - \frac{{a_3 }}
{s}} \right)^{\frac{1}
{2} + q - m} \left( {1 - \frac{{a_4 }}
{s}} \right)^{m - 1} . \cr
\end{split}
\end{equation}
Taking into account that the PDF is related to the MGF by the inverse Laplace transform, i.e., 
$f_{\gamma}(x)=\L^{-1}[M_{\gamma}(-s);x]$,  (\ref{eq:36}) follows from (\ref{eq:40}) and the
Laplace transform pair given in [17, eq. (4.24.3)]. On the other hand, (\ref{eq:37}) is obtained analogously by considering that  $F_{\gamma}(x)=\L^{-1}[M_{\gamma}(-s)/s;x]$.

\section{Proof of Lemma 3}
\label{App4}

The exact expression for the TWDP fading power envelope PDF has integral form \cite{Durgin2002}
\begin{align}
\label{eq:app01}
&f^{\text{T}}_{\gamma}(\gamma)=\beta\exp\left\{-\beta\gamma\right\}\exp\left\{-K\right\}\frac{1}{\pi}\\
&\times \int_{0}^{\pi}\exp\left\{-K\Delta\cos\theta\right\}I_{0}\left(2\sqrt{{\gamma}\beta K(1-\Delta\cos\theta)}\right)d\theta\nonumber,
\end{align}
where $\beta=(1+K)/\bar\gamma$. In order to circumvent this issue, a family of PDFs that approximate the exact PDF of the TWDP fading model was given in \cite[eq. (17)]{Durgin2002}. These approximate PDFs for the TWDP power envelope are expressed in closed-form as:
\begin{equation}
\label{eq:app02}
\begin{split}
\widehat f_{\gamma}^{\text{T}}(\gamma)\approx  \sum_{i=1}^{M}\frac{\alpha_i}{2}  & \left\{\mathcal{F}\left(\gamma;\beta,K(1-\delta_i)\right) \right. \\&
+\left. \mathcal{F}\left(\gamma;\beta,K(1+\delta_i)\right)\right\},
\end{split}
\end{equation}
where  
\begin{align}
\label{eq:app03}
\mathcal{F}\left(\gamma;\beta,K\right)\triangleq \beta \exp (-\beta\gamma)\exp(-K) I_0\left(2\sqrt{\gamma\beta K}\right),
\end{align}
and the coefficients $\alpha_i$ and $\delta_i$ are given by
\begin{align}
\label{coef1}
\alpha_i&=\frac{2(-1)^i}{(2M-1)(2M-i)!(i-1)!}
\\& \ \ \ \nonumber \times \int_{0}^{2M-1}\prod_{\stackrel{k=1}{k \neq i}}^{2M}\left(u-k+i\right)du,\\
\label{coef2}
\delta_i&=\Delta\cos\left(\frac{(i-i)\pi}{2M-1}\right),
\end{align}
respectively. The number of terms in the summation is related to the values of $K$ and $\Delta$; as argued in \cite{Durgin2002}, setting $M> K\Delta$ suffices to closely match the exact PDF in (\ref{eq:app01}).

Since (\ref{eq:app03}) corresponds to the PDF of a Rician power envelope, expression (\ref{eq:app02}) allows for approximating the TWDP distribution in terms of a mixture of $2M$ Rician distributions.

Following a similar reasoning as in Appendix \ref{App1}, the approximate PDF for the FTR fading power envelope can be obtained by averaging (\ref{eq:app02}) with $K_u=uK$ over all possible realizations $u$ of the random variable $\zeta$, which follows a Gamma distribution as indicated in (\ref{eq:03}). Thus, the mixture of Rician PDFs in (\ref{eq:app02}) averaged over a Gamma distribution leads to the following approximate expression of the FTR power envelope PDF:
\begin{equation}
\begin{split}
\label{eq:app05}
\widehat f_{\gamma}(\gamma)\approx \sum_{i=1}^{M}  \frac{\alpha_i}{2} & \left\{\mathcal{G}_m\left(\gamma;\beta,K(1-\delta_i)\right) \right.
\\& \left. +\mathcal{G}_m\left(\gamma;\beta,K(1+\delta_i)\right)\right\},
\end{split}
\end{equation}
where
\begin{align}
\label{eq:app06}
\mathcal{G}_m\left(\gamma;\beta,K\right)&\triangleq \int_0^{\infty}\mathcal{F}\left(\gamma;\beta,uK\right)f_{\zeta}(u)du\\
&= \frac{{m^m }}{{\Gamma \left( m \right)}} \int_0^{\infty}\mathcal{F}\left(\gamma;\beta,uK\right)u^{m-1}e^{ - mu}du\nonumber\\
& = \frac{m^m \beta}{(K+m)^m}e^{-\beta\gamma}{}_1F_1\left(m,1;\frac{K\beta\gamma}{K+m}\right)\nonumber.
\end{align}
where ${}_1F_{1}(\cdot,\cdot;\cdot)$ is the Kummer confluent hypergeometric function, and the same steps as in \cite[App. A]{Paris2014} have been used to derive the last equation. Thus, the FTR fading power envelope PDF corresponds to a mixture of $2M$ Rician shadowed PDFs \cite{Abdi2003}, which is in coherence with the connection between the TWDP and Rician distributions that exists in the absence of the additional fluctuation in the specular waves here considered.

Finally, noting that for $m\in\mathbb{Z}^+$ the Kummer hypergeometric function can be expressed in terms of the Laguerre polynomials by using \cite[eq. 24]{erdelyi1940} and the well-known Kummer transformation, we have
\begin{align}
\label{eq:app07}
{}_1{F}_1(m,1;z)=e^z \sum_{n=0}^{m-1}\binom{m-1}{n}\frac{z^n}{n!}.
\end{align}
Thus, combining (\ref{eq:app05})-(\ref{eq:app07}) yields the closed-form approximation for the PDF of the FTR fading power envelope in (\ref{eq:36b}), in terms of a finite sum of exponential functions and powers. Direct integration of (\ref{eq:app05}) yields the approximate expression for the FTR fading power envelope CDF in (\ref{eq:37b}). 

Strikingly, we must note that the additional fluctuation introduced by the FTR fading model does not cause any increase in mathematical complexity, but instead facilitates the mathematical tractability.

\section{Proof of Lemma 4}
\label{App5}
The MGF of $\gamma$ can be found as
\begin{equation}
\label{eq:56}
\begin{split}
M_\gamma  \left( s \right) = \int_0^\infty  {\int_0^\infty  {M_{\gamma _{_{u_1 ,u_2 } } } \left( s \right)} f_{\zeta _1 } \left( {u_1 } \right)f_{\zeta _2 } \left( {u_2 } \right)du_1 du_2 } ,
\end{split}
\end{equation}
where $f_{\zeta _i } \left(  \cdot  \right),\;i = 1,2$, is given by (\ref{eq:03}) and
\begin{equation}
\label{eq:57}
\begin{split}
  & M_{\gamma _{_{u_1 ,u_2 } } } \left( s \right) = {{1 + K_{u_1 ,u_2 } } \over {1 + K_{u_1 ,u_2 }  - \bar \gamma _{_{u_1 ,u_2 } } s}}  \cr 
  &   \cdot \exp \left( {{{K_{u_1 ,u_2 } \bar \gamma _{u_1 ,u_2 } s} \over {1 + K_{u_1 ,u_2 }  - \bar \gamma _{u_1 ,u_2 } s}}} \right)I_0 \left( {{{\Delta _{u_1 ,u_2 } K_{u_1 ,u_2 } \bar \gamma _{u_1 ,u_2 } s} \over {1 + K_{u_1 ,u_2 }  - \bar \gamma _{u_1 ,u_2 } s}}} \right), \cr
\end{split}
\end{equation}
with
\begin{equation}
\label{eq:58a}
\begin{split}
\bar \gamma _{u_1 ,u_2 }  = \left( {E_b /N_0 } \right)\left( {u_1 V_1^2  + u_2 V_2^2 ` + 2\sigma ^2 } \right),
\end{split}
\end{equation}
\begin{equation}
\label{eq:58}
\begin{split}
K_{u_1 ,u_2 }  = {{u_1 V_1^2  + u_2 V_2^2 } \over {2\sigma ^2 }},
\end{split}
\end{equation}
\begin{equation}
\label{eq:59}
\begin{split}
\Delta _{u_1 ,u_2 }  = {{2\sqrt {u_1 u_2 } V_1 V_2 } \over {u_1 V_1^2  + u_2 V_2^2 }}.
\end{split}
\end{equation}
Using a similar approach as the one in Appendix \ref{App1}, the double integral in (\ref{eq:56}) can be solved in closed-form with the help of \cite[p. 197 (20)]{Erdelyi1854} and \cite[p. 215 (11)]{Erdelyi1854}, yielding (\ref{eq:54}).

\bibliographystyle{IEEEtran}
\bibliography{bibjavi}

\quad

\end{document}